\documentclass[12pt,a4paper,fleqn]{article}

\usepackage{lscape,exscale,amsthm}

\usepackage[intlimits]{amsmath}
\usepackage{algorithm}
\usepackage{rawfonts}
\usepackage{latexsym}
\usepackage[cp850]{inputenc}
\usepackage{epsfig}
\usepackage{bbm,amssymb}
\usepackage{enumerate,dsfont}
\usepackage{natbib}
\newcommand{\btheta}{\mbox{\boldmath$\theta$}}

\newtheorem{theorem}{Theorem}

\usepackage{authblk}
\title{Bayesian model selection: Application to adjustment of fundamental physical constants}

\author[1]{Olha Bodnar}
\author[2]{Viktor Eriksson}

\affil[1]{Unit of Statistics, School of Business, \"Orebro University, SE-70182 \"Orebro, Sweden}
\affil[2]{Department of Statistics, Uppsala University, Sweden, SE-75236 Uppsala, Sweden}

\providecommand{\keywords}[1]
{
\small	
\textbf{\textit{Keywords}:} #1
}
\begin{document}

%
\maketitle
\begin{abstract}
The location-scale model is usually present in physics and chemistry in connection to the Birge ratio method for the adjustment of fundamental physical constants such as the Planck constant or the Newtonian constant of gravitation, while the random effects model is the commonly used approach for meta-analysis in medicine. These two competitive models are used to increase the quoted uncertainties of the measurement results to make them consistent. The intrinsic Bayes factor (IBF) is derived for the comparison of the random effects model to the location-scale model, and we answer the question which model performs better for the determination of the Newtonian constant of gravitation. The results of the empirical illustration support the application of the Birge ratio method which is currently used in the adjustment of the CODATA 2018 value for the Newtonian constant of gravitation together with its uncertainty. The results of the simulation study illustrate that the suggested procedure for model selection is decisive even when data consist of a few measurement results.
\end{abstract}

\vspace{1cm}

\newpage
\section{Introduction}\label{sec:intro}

Fundamental constants in physics and chemistry, like the Newtonian constant of gravitation or the Planck constant, are usually determined by results of several studies which are performed at different times and places (\citet{codata2016}, \citet{newell2018codata}, \citet{alighanbari2020precise}). In many applications in medicine the consensus value is obtained by pooling the results of clinical studies together (\citet{brockwell2001}, \citet{Lambert2005}, \citet{Higgins2009}, \citet{Rukhin2013}, \citet{bodnar2017bayesian}, \citet{jones2018use}), while interlaboratory comparison study is one of the most important topics in metrology (\citet{mandel-1970}, \citet{rukhin-2003}, \citet{bodnar2020assessing}).

Each individual study reports an estimate of the quantity of interest together with its uncertainty which are based, for example, on measurements obtained in laboratories in physics and chemistry where a great care is taken to determine the measurement margin of error such that it accounts for every identifiable source of error. For that reason, it might therefore be assumed that an individual study estimate is internally consistent, i.e., its uncertainty takes into account of every identifiable (hence explainable) source of uncertainty.

Although the individual studies are internally consistent, they often fail to be externally consistent. More precisely, when the estimates together with their uncertainties from all individual studies are brought together in a scatterplot, it is usually revealed that the dispersion of the estimates is substantially greater than what the reported uncertainties suggest them to be. Such a situation is depicted in Figure \ref{fig:G} where the measurements of the Newtonian constant of gravitation G are presented together with reported uncertainties. Most of the uncertainties provided by laboratories are not able to capture the variation observed in the scatterplot.

\begin{figure}[H]
	\begin{center}
		\includegraphics[scale=0.6, clip=true, trim=0 10 0 10]{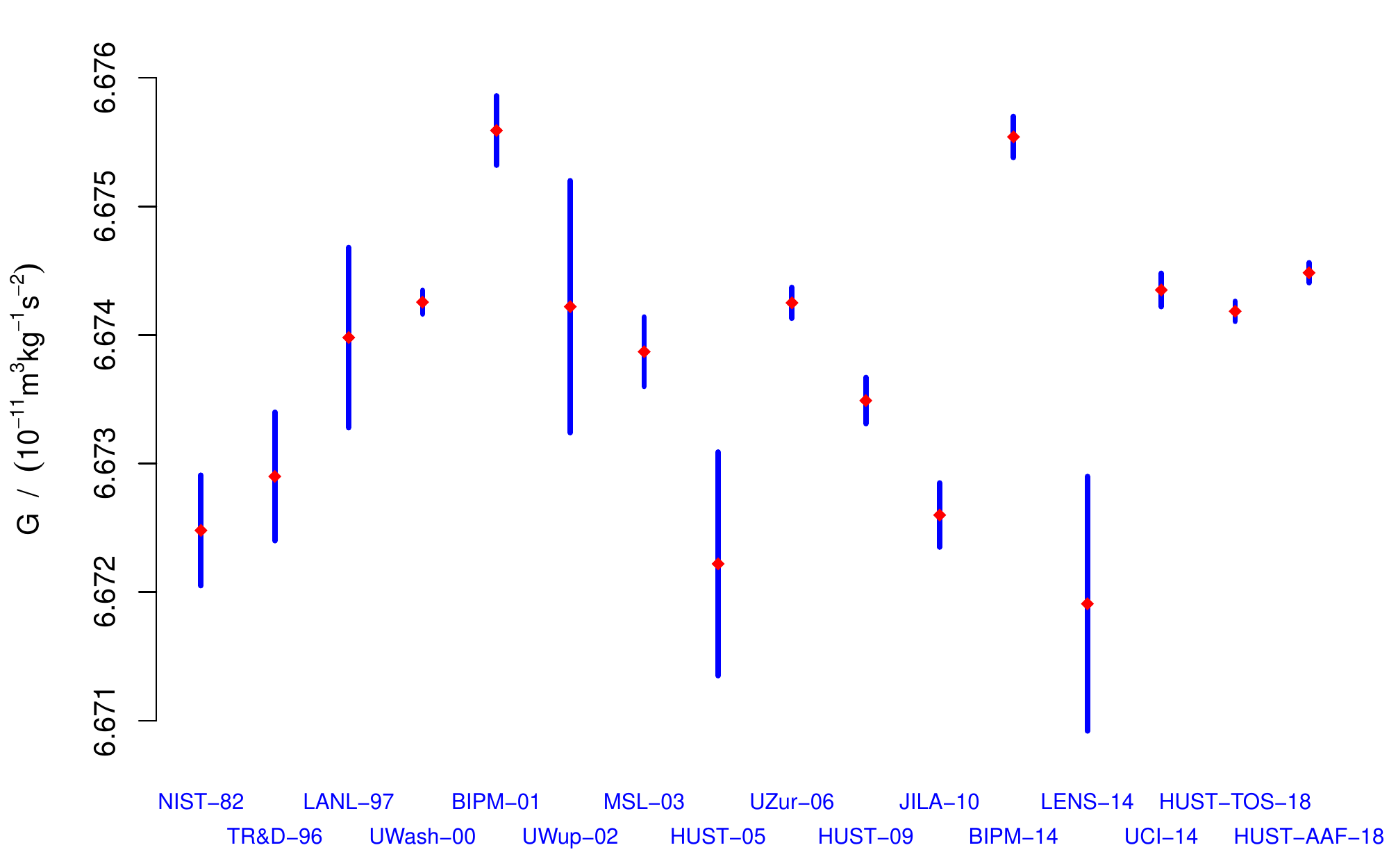}
	\end{center}
 \caption{Newtonian constant of gravitation G. The red diamonds are the measurement results and the blue intervals are the uncertainties provided by each laboratory (data are from Figure 1 in \citet{bodnar2020bayesian}).}
 \label{fig:G}
\end{figure}

The reported uncertainties take into account every identifiable source of error in each individual study, yet these uncertainties fail to explain the full variation between the estimates satisfactory. The extra variation is caused by external inconsistency, i.e., inconsistency between the laboratories or interlaboratory inconsistency. This extra variation is usually labeled as the heterogeneity and it also be referred to as the excess variance or as the dark uncertainty (see, \citet{thompson2011dark}). It is remarkable that the heterogeneity is unexplainable and cannot be resolved easily.

Birge ratio method and random effects meta-analysis seem be the two most popular approaches to account for the dark uncertainty when the results of several studies are pooled together. Birge ratio method is generally favored by physicists and chemists in parameter estimation (\citet{codata2016}). It is related to the Birge ratio introduced by \citet{Birge1932} and is used to adjust for inconsistent data in \citet{weise2000removing} and \citet{codata2016}, among others. The random effects model is generally favored by medical researchers in parameter estimation (see, \citet{hardy1998detecting}, \citet{AdesHiggins2005}, \citet{Turner2015}, \citet{guolo2017random}, \citet{veroniki2019methods}). It has been used in chemistry for interlaboratory comparisons in \citet{mandel-1970} and in \citet{TomanFischerElster2012}. In \citet{BodnarLinkElster2015} objective Bayesian inferences are established for the generalized random effects model and are used to estimate the Planck constant.

The random effects model has been compared to the Birge ratio methods presented as the location-scale model in \citet{bodnar2016evaluation}. The comparison of two models was performed in terms of robustness analysis to model misspecification by computing coverage intervals for the overall mean. The results of the numerical study reveal that the random effects model is more robust to model misspecification than the location-scale model. Namely, when the data were generated from the location-scale model but the random effects model was used to estimate the overall mean, then the coverage probability was roughly 95\%, the significance level used in the simulation study. However, when the data were drawn from the random effects model but the location-scale model was used to estimate the overall mean, then the coverage probability was found to be roughly 0.45\%. As a result, one can remark that if the location-scale model is actually true, then the random effects model can still do a good job at estimating the overall mean, while if the random effects model is actually true, then the application of the location-scale model can lead to unreliable inference for the overall mean. Finally, it can be noted that both models have good coverage probabilities given that the used model is the true one.

The results of \citet{bodnar2016evaluation} present only one aspect of model selection procedures and they do not provide an answer to the question which model is more preferable to the observed data. In the paper we contribute to this challenging task by developing a new approach to distinguish between the two models for dark uncertainty from the viewpoint of objective Bayesian statistics by performing model selection based on intrinsic Bayes factor. There are several reasons to opt for Bayesian model selection (see, e.g., \citet[p.138-140]{berger2001objective}: (i) the Bayes factor is interpreted as an odds factor which is easy to understand; (ii) it is an automatic Ocham's razors since it favor a simple model over a complex model; (iii) the Bayes factor is a consistent method, which means that if any of the compared models is true, then the true model is selected (given some mild conditions are satisfied); (iv) it does not require models to be nested, which makes it possible to compare different types of models in which the parameters may vary in dimension. Some of the difficulties with the Bayes factor is also stated in \citet[p.142]{berger2001objective}. For example, the Bayes factor can be difficult to compute. Also, the Bayes factor does not usually return good answers when we use vague proper priors leading to the application of improper noninformative priors.

Traditionally, the Bayes factor is computed by endowing the model parameters with an informative prior. On the other hand, no information or only vague information about model parameters can be present in practice (see, e.g., \citet{Lambert2005}). As a result, the influence of the prior distribution on the Bayes factor is greater than it is on the estimation of a parameter. This is because the weight of the prior distribution is not reduced with an increasing sample size in the Bayes factor (\cite{berger1996intrinsic}). The Bayes factor is therefore particularly vulnerable to misspecifications in the prior distribution and for this reason we look for an objective noninformative prior, like a Laplace prior (\citet{Laplace1812}), Jeffreys prior (\citet{Jeffreys1946}) or reference prior (\citet{BergerBernardo1992c}) which is obtained by maximizing the Shannon mutual information between the prior and posterior, thus reducing the impact of the prior on posterior (see, \citet{BergerBernardo1992c}, \citet{ClarkeYuan2004}, \citet{berger2009formal}, \citet{bodnar2014analytical}). The Berger and Bernado reference prior is recommended for use in the computation of intrinsic Bayes factor for large sample sizes  (\citet[p.121]{berger1996intrinsic}).

There are various approaches for objective Bayesian model selection, however the methods usually have some limitations or difficulties in their implementations. The traditional Bayes factor computed by endowing an informative proper prior cannot be longer used for objective Bayesian model selection based on a noninformative priors which are usually improper. A modification of the definition of the conventional Bayes factor is required. Several approaches exist in the literature with the ones proposed by \citet{o1995fractional}) and \citet{berger1996intrinsic} seem to be the ones which are mostly used. \citet{o1995fractional} suggested the application of the fractional Bayes factor for the model selection when an improper prior is used, while \citet{berger1996intrinsic} developed the theory of the Bayesian model selection based on the intrinsic Bayes factor. Moreover, \citet[p.121]{berger1996intrinsic} recommended the use of the Berger and Bernado reference prior in the computation of intrinsic Bayes factor for large sample sizes. The benefit with the intrinsic Bayes factor is that it is automatic (objective) in the sense that it only depends on observed data and noninformative prior distribution. It may, however, be computationally intensive and unstable (see, \citet[p.120-121]{berger1996intrinsic}). The intrinsic Bayes factor makes use of a minimal training sample for parameter estimation and in order to make for a stable result it is conventional to take the average or median of the intrinsic Bayes factors computed over all possible training samples. The empirical probability for the Bayesian model selection based on the intrinsic Bayes factor will be introduced in this paper and applied in the empirical illustration about the Newtonian constant of gravitation.

The rest of the paper is organized as follows. In Section \ref{sec:BMS}, we introduce the concept of Bayesian model selection with the special emphasis on the case when the model parameters are endowed with a noninformative prior. The competitive models, namely the location-scale model related to the Birge ratio approach and the random effects model, are presented in Section \ref{sec:model}. Here, we also derive the quantities needed in the computation of the intrinsic Bayes factor. In Section \ref{sec:sim} we provide a numerical comparison of the procedures in the small sample case, while it is implemented to the measurements of the Newtonian constant of gravitation in Section \ref{sec:emp}. Discussion of the obtained results are given in Section \ref{sec:sum}, while the derivation of theoretical results is presented in the appendix (Section \ref{sec:app}).

\section{Bayesian model selection} \label{sec:BMS}

In this section we introduce the Bayes factor, provide its interpretation, and discuss its use for model comparison. Let $M_1, \dots,M_i,\cdots, M_q$ denote the models that we aim to make a comparison between. Under model $M_i$ the data $\mathbf{x}$ follow the probability density function $f_i(\mathbf{x}|\btheta_i)$ with $\btheta_i\in\Theta_i$, which is also referred to as the likelihood function. In the Bayesian approach to model selection we start by assigning for each model $M_i$ a prior probability $P(M_i)$ of the model $M_i$ being correct and a prior distribution $\pi_i(\btheta_i)$ of the parameter $\btheta_i$. Then the posterior probability $P(M_i|\mathbf{x})$ of the model $M_i$ being correct when the sample $\mathbf{x}$ is observed is given by
\begin{equation}
P(M_i|\mathbf{x})=\left(\sum_{j=1}^{q}\frac{P(M_j)}{P(M_i)}B_{ji}(\mathbf{x})\right)^{-1} , \nonumber
\end{equation}
where
	\begin{equation}
	B_{ji}(\mathbf{x})=\frac{m_j(\mathbf{x})}{m_i(\mathbf{x})}=\frac{\int f_j(\mathbf{x}|\btheta_j)\pi_j(\btheta_j)d\btheta_j}{\int f_i(\mathbf{x}|\btheta_i)\pi_i(\btheta_i)d\btheta_i}  \nonumber
	\end{equation}
is the Bayes factor between $M_j$ and $M_i$ and
\begin{equation}
m_i(\mathbf{x})=\int_{\btheta_i \in \Theta_i} f_i(\mathbf{x}|\btheta_i)\pi_i(\btheta_i)d\btheta_i  \nonumber
\end{equation}
is the marginal distribution of the data under model $M_i$, which describes the probability of observing $\mathbf{x}$ under model $M_i$. 

The posterior probability $P(M_i|\mathbf{x})$ is dependent on subjectively chosen prior probabilities $P(M_i)$ for $i=1,2,\dots,q$ of the model $M_i$ being correct. If one prefers the Bayesian model selection procedure to be objective, then she/he shall assign equal prior probabilities to the considered models, that is $P(M_1)=...=P(M_q)=1/q$. In this case the posterior probability $P(M_i|\mathbf{x})$ used for the selection of a model will depend on the Bayes factors $B_{ji}(\mathbf{x})$ only. In a special case, when $q=2$, the Bayes factor $B_{ji}(\mathbf{x})$ can alone be used to perform Bayesian model selection. The Bayes factor is interpreted as the odds in favor of $M_j$ against $M_i$ based on the evidence in the data $\mathbf{x}$. The inequality $B_{ji}>1$ is supportive of model $M_j$, while $B_{ji}<1$ is interpreted as evidence against $M_j$. 
Finally, it is noted that the prior distributions $\pi_i(\btheta_i)$ of the parameter $\btheta_i$ are subjective because they are decided prior to the data, but they are objective if one will use a noninformative prior.

The Bayes factor can also be interpreted as an odds factor if the marginal distribution $m_i(\mathbf{x})$ is proper, such that $\int_{\mathbf{x\in\Omega}}m_i(\mathbf{x})d\mathbf{x}=1$ for all $i=1,\dots,q$. 
However, the latter equality may not always be fulfilled when the model parameters are endowed with an improper prior, denoted by $\pi_i^N(\btheta_i)$, i.e., when objective Bayesian inference are preferable. An improper prior distribution $\pi_i^N(\theta_i)$ cannot have a normalizing constant and its application leads to the improper marginal distribution of data $m_i(\mathbf{x})$. As a result, one cannot longer apply the conventional Bayesian model selection without a modification of the definition of the Bayes factor. In the next section, we describe the intrinsic Bayes factor developed by \citet{berger1996intrinsic} for Bayesian model selection with improper noninformative priors.

\subsection{Intrinsic Bayes factor}\label{sec:IBF}

The intrinsic Bayes factor (IBF) is a solution to the problem that is imposed on the Bayes factor by the improper prior distribution, namely that the Bayes factor has no interpretation as an odds factor under an improper prior distribution. The IBF differs from the ordinary Bayes factor in that it makes use of a training sample. The training sample is employed for the purpose of transforming the improper prior distribution to the proper posterior distribution, which is then used as a prior for the rest of the elements in the sample.

The training sample $\mathbf{x}_{\ell}$ is a subset of $\mathbf{x}$ of size $m$. Let $\mathbf{x}_{(\ell)}=\mathbf{x}-\mathbf{x}_{\ell}$ be the data set with the training sample removed. Then $\mathbf{x}_{\ell}$ and $\mathbf{x}_{(\ell)}$ forms a partition of $\mathbf{x}$. A training sample, $\mathbf{x}_{\ell}$, is called proper if $0<\pi_i^N(\btheta_i|\mathbf{x}_{\ell})<\infty$ for all $M_i$, and minimal if it is proper and no subset of $\mathbf{x}_{\ell}$ is proper. The notation $\pi_i^N(\btheta_i|\mathbf{x}_{\ell})$ stands for the posterior distribution of $\btheta_i$ computed based on the noninformative prior $\pi_i^N(\btheta_i)$ and sample $\mathbf{x}_{\ell}$. As a result, the training sample is used to transform the improper prior $\pi_i^N(\btheta_i)$ to the proper posterior $\pi_i^N(\mathbf{\theta_i|x}_{\ell})$, which is then used in the computation of the conditional distribution of $\mathbf{x}_{(\ell)}$ given the training sample $\mathbf{x}_{\ell}$. It is noted that a part of data, that could have been used for computing the Bayes factor, is lost since it is needed for the specification of the proper posterior $\pi_i^N(\btheta_i|\mathbf{x}_{\ell})$. We are sacrificing a portion of the data set $\mathbf{x}$, namely the training sample $\mathbf{x}_{\ell}$, for parameter estimation instead of using it for model selection. On the other side, we want as much data as possible for the computation of the Bayes factor since the overall goal is to make a comparison between two models. This is why the minimal training sample is preferable.

The minimal training sample is usually of the same size as the number of parameters in the model. But there are exceptions to this rule, if for example the prior distribution is proper to some parameters then the minimal training sample size may be less then the number of parameters in the model (see, \citet{berger1996intrinsic}). As an example, suppose we have a joint prior distribution of two parameters that is improper. Then only two observations from the data set should suffice as a training sample. When the sample size is $n$, then there are $n(n-1)/2$ combinations of choice to select two observations from $\mathbf{x}$. One should pick every combination to compute the IBF, which implies $n(n-1)/2$ computations of the IBF. Then the conclusions are drawn based on a summary of the differently computed IBFs, such as the average IBF or the median IBF.

Let $m_i^N(\mathbf{x}_{\ell})$ be the marginal distribution of $\mathbf{x}_{\ell}$ computed under the noninformative prior $\pi_i^N(\btheta_i)$ and let $m_i^N(\mathbf{x}_{(\ell)}|\mathbf{x}_{\ell})$ denote the conditional distribution of $\mathbf{x}_{(\ell)}$ given $\mathbf{x}_{\ell}$. Then the intrinsic Bayes factor computed for the sample $\mathbf{x}_{(\ell)}$ conditioned on the training sample $\mathbf{x}_{\ell}$ is given by
\begin{equation*}
B_{ji}^I(\mathbf{x}_{(\ell)}|\mathbf{x}_{\ell})=\frac{m_j^N(\mathbf{x}_{(\ell)}|\mathbf{x}_{\ell})}{m_i^N(\mathbf{x}_{(\ell)}|\mathbf{x}_{\ell})}=\frac{\int f_j(\mathbf{x}_{(\ell)}|\theta_j,\mathbf{x}_{\ell})\pi^N_j(\theta_j|\mathbf{x}_{\ell})d\theta_j}{\int f_i(\mathbf{x}_{(\ell)}|\theta_i,\mathbf{x}_{\ell})\pi^N_i(\theta_i|\mathbf{x}_{\ell})d\theta_i} \ , \nonumber
\end{equation*}
which can also be rewritten as
\begin{eqnarray*}
B_{ji}^I(\mathbf{x}_{(\ell)}|\mathbf{x}_{\ell})&=&\frac{m_j^N(\mathbf{x}_{(\ell)},\mathbf{x}_{\ell})}{m_j^N(\mathbf{x}_{\ell})}\bigg/\frac{m_i^N(\mathbf{x}_{(\ell)},\mathbf{x}_{\ell})}
{m_i^N(\mathbf{x}_{\ell})} \\
&=&\frac{m_j^N(\mathbf{x})}{m_j^N(\mathbf{x}_{\ell})}\bigg/\frac{m_i^N(\mathbf{x})}{m_i^N(\mathbf{x}_{\ell})} = \frac{m_j^N(\mathbf{x})}{m_i^N(\mathbf{x})}\times \frac{m_i^N(\mathbf{x}_{\ell})}{m_j^N(\mathbf{x}_{\ell})} =B_{ji}^N(\mathbf{x})\times B_{ij}^N(\mathbf{x}_{\ell}).
\end{eqnarray*}

Recall that the challenge with the conventional Bayes factor was related to improper prior distributions. For two models $M_j$ and $M_i$ with likelihoods $f_j(\mathbf{x}|\btheta_j)$ and $f_i(\mathbf{x}|\btheta_i)$ and with employed improper priors $\pi_j^N(\btheta_j)=c_jh_j(\btheta_j)$ and $\pi_i^N(\btheta_i)=c_ih_i(\btheta_i)$, the IBF of $M_j$ and $M_i$ is given by
\begin{equation}
\begin{aligned}
&B_{ji}^I(\mathbf{x}_{(\ell)}|\mathbf{x}_{\ell})=B_{ji}^N(\mathbf{x})\times B_{ij}^N(\mathbf{x}_{\ell}) \\ \\
&=\frac{c_j}{c_i}\frac{\int f_j(\mathbf{x}|\theta_j)h_j(\theta_j)d\theta_j}{\int f_i(\mathbf{x}|\theta_i)h_i(\theta_i)d\theta_i} \times \frac{c_i}{c_j}\frac{\int f_i(\mathbf{x}_{\ell}|\theta_i)h_i(\theta_i)d\theta_i}{\int f_j(\mathbf{x}_{\ell}|\theta_j)h_j(\theta_j)d\theta_j} \ . \nonumber
\end{aligned}
\end{equation}
The important point is that the arbitrary constants of proportionality $c_j$ and $c_i$ will cancel in the IBF (see, \citet{o1995fractional}). Hence the IBF does not depend on anything arbitrary or subjective and can therefore be used for an objective Bayesian model selection.

We summarize this subsection by saying that a noninformative prior distribution is used for an objective approach to Bayesian model selection. The IBF is a method in which the noninformative improper prior $\pi_i^N(\btheta_i)$ is transformed to the proper posterior $\pi_i^N(\btheta_i|\mathbf{x}_{\ell})$ by the use of a minimal training sample $\mathbf{x}_{\ell}$, while the rest of the data $\mathbf{x}_{(\ell)}$ is employed for model selection resulting in the IBF. The IBF measures the odds in favor of $M_j$ against $M_i$ based on evidence in the data, without the arbitrariness of the constants $c_j$ and $c_i$.

\section{Models for dark uncertainty} \label{sec:model}
In this section we will establish all theoretical results that are needed to compute the IBF for the location-scale model and the random effects model. Both of these models are used in the determination of physical constants and in the meta-analysis to form a consensus value while accounting for an unexplainable heterogeneity. We first present the location-scale model and derive the results needed for the calculation of the IBF. Then, we treat the random effects model likewise. At the end of this section the IBF used for the comparison of these two models is provided.

\subsection{Location-scale model} \label{sec:LSM}
The location-scale model is related to the Birge method, which is a popular method in physics, chemistry, and metrology for adjusting uncertainties to adapt for external inconsistency (see, \citet{Birge1932}, \citet{codata2016}). The Birge method is used under the assumption that due to external inconsistency, i.e., heterogeneity, the reported uncertainty for each estimate is understated. The Birge method multiplies each uncertainty by a common factor that is called the Birge ratio which is a sort of averaging of uncertainties. As a result all the uncertainties are adjusted to be uniformly understated and hence mutually consistent. In \citet{bodnar2014adjustment} the Birge ratio is studied within a Bayesian framework in relation to the location-scale model.

Let $\mathbf{V}$ be a known positive definite matrix, $\mu$ an unknown location parameter and $\tau_{LS}$ an unknown scale parameter. Then the location-scale model $M_{LS}$ of the data $\mathbf{x}$ is defined by the following two equations:
\begin{eqnarray}\label{model:LSM}
\mathbf{x}&=&\mu\mathbf{1}+\boldsymbol{\varepsilon}_{LS} \quad \text{with} \quad \boldsymbol{\varepsilon} \sim N(\mathbf{0}, \sigma_{LS}^2\mathbf{U}).
\end{eqnarray}

The location-scale model assumes that the measurement data are obtained from a multivariate normal distribution with the mean $\mu\mathbf{1}$ and the covariance matrix $\tau_{LS}^2\mathbf{U}$, i.e., $\mathbf{x}|\mu, \tau_{LS}, M_{LS}\sim N(\mu\mathbf{1},\tau_{LS}^2\mathbf{U})$, where $\mathbf{1}$ stands for the $n$-dimensional vector of ones. The positive definite matrix $\mathbf{U}$ contains all the reported uncertainties and covariances of the measurement results. Namely, the diagonals of $\mathbf{U}$ contain the variances $u_i^2$ of the estimates and the off-diagonal entries contain covariances $\rho_{ij}u_iu_j$, where $\rho_{ij}$ is the correlation between the $i^{th}$ and the $j^{th}$ study. The factor $\tau_{LS}$ is an unknown quantity that represents the unexplainable interlaboratory heterogeneity. The overall mean $\mu$ is the target parameter in many applications in physics, chemistry and medicine, while $\tau_{LS}$ is a nuisance parameter included in the model to capture the heterogeneity. The motivation behind the model application is that the elements in $\mathbf{U}$ can be understated due to the presence of heterogeneity. By multiplying $\mathbf{U}$ by $\tau_{LS}^2$, the elements in $\mathbf{U}$ then become uniformly adjusted for by the heterogeneity. On the other side, the reported correlations $\rho_{ij}$ become unchanged under this adjustment. The notation $\mathbf{x}|\mu, \tau_{LS}, M_{LS}$ makes it explicit that the data are conditioned on the two parameters $\mu$ and $\tau_{LS}$ under the location-scale model $M_{LS}$.

The likelihood function of $\mathbf{X}$, when the location-scale model is assumed, is given by
\begin{equation}\label{lik:LSM}
	f(\mathbf{x}|\mu, \tau_{LS}, M_{LS})=\frac{\tau_{LS}^{-n}}{(2\pi)^{n/2}} \mathsf{exp}\bigg(-\frac{1}{2\tau_{LS}^{2}}(\mathbf{x}-\mu\mathbf{1})^{T}\mathbf{U}^{-1}(\mathbf{x}-\mu\mathbf{1})\bigg).
\end{equation}
The Berger \& Bernado reference prior has been derived for the general location-scale model in \citet{FernandezSteel1999} and is given by
\begin{equation}\label{prior:LSM}
\pi^N(\tau_{LS})=\pi^N(\mu,\tau_{LS})\propto \frac{1}{\tau_{LS}} \ .
\end{equation}

The marginal distribution of $\mathbf{x}$ and of the training sample $\mathbf{x}_{\ell}$ are derived in Theorem \ref{th1} whose proof is given in the appendix (see, Section \ref{sec:app}).

\begin{theorem}\label{th1}
Let $n>2$ and let $\mathbf{U}$ be positive definite. Then under the location-scale model \eqref{model:LSM} and reference prior \eqref{prior:LSM},
\begin{enumerate}[(i)]
\item the marginal distribution of the whole sample is given by
\begin{equation}\label{m:LSM}
	m(\mathbf{x}|M_{LS})=\frac{\Gamma(\frac{n-1}{2})\big(\mathbf{x}^T\mathbf{Q}\mathbf{x}\big)^{-\frac{n-1}{2}}}{\big(\mathsf{det}(\mathbf{U})\big)^{1/2} 2\pi^{\frac{n-1}{2}}\sqrt{\mathbf{1}^T\mathbf{U}^{-1}\mathbf{1}}} \ ,
\end{equation}
where
\begin{equation}\label{Q:LSM}
\mathbf{Q}= \mathbf{U}^{-1}-\frac{\mathbf{U}^{-1}\mathbf{1}\mathbf{1}^T\mathbf{U}^{-1}}{\mathbf{1}^T\mathbf{U}^{-1}\mathbf{1}}
\end{equation}
\item The size of a minimal training sample $\mathbf{x}_{\ell}$ is given by $m=2$, i.e., $\mathbf{x}_{\ell}=\{x_i,x_j\}$, where $x_i,x_j\in\mathbf{x}$ and $i\neq j$. Moreover, the marginal distribution of the training sample is given by
\begin{equation}\label{ml:LSM}
m(\mathbf{x}_{\ell}|M_{LS})=\frac{1}{2\sqrt{(\mathsf{det}(\mathbf{U}_{\ell}))(\mathbf{x}_{\ell}^T\mathbf{Q}_{\ell}\mathbf{x_{\ell}})(\mathbf{1}_2^T\mathbf{U}_{\ell}^{-1}\mathbf{1}_2)}},
\end{equation}
where $\mathbf{U}_{\ell}$ is obtained from $\mathbf{U}$ by taking its elements lying on the intersections of the $\ell$ rows and $\ell$ columns,
\begin{equation}\label{Ql:LSM}
	\mathbf{Q}_{\ell}=\mathbf{U}^{-1}_{\ell}-\frac{\mathbf{U}_{\ell}^{-1}\mathbf{1}_2\mathbf{1}_2^T\mathbf{U}_{\ell}^{-1}}{\mathbf{1}_2^T\mathbf{U}_{\ell}^{-1}\mathbf{1}_2}
	\quad \text{with} \quad \mathbf{1}_2=(1,1)^T .
\end{equation}
\end{enumerate}	
\end{theorem}

The results of Theorem \ref{th1} provide all quantities which are need for the computation of the IBF as shown in Section \ref{sec:IBF-LSM-REM}.

\subsection{Random effects model} \label{sec:REM}
The random effects model is a classical model that has been studied from both frequentist statistics (\citet{Cochran37}, \citet{Cochran54}, \citet{YatesCochran38}, \citet{Rao1997}, \citet{Searle2006}) and Bayesian statistics (\citet{Hill65}, \citet{TiaoTan1965}, \citet{BrowneDraper06}, \citet{Gelman06}, \citet{BodnarLinkElster2015}). In metrology it has been considered in \citet{Kacker2004}, \citet{TomanFischerElster2012}, while it is one of the mostly used in performing meta-analyses in medicine (see, \citet{Lambert2005}, \citet{Turner2015}, \citet{veroniki2019methods}).

The random effects model $M_{RE}$ is defined by
\begin{eqnarray}\label{model:REM}
\mathbf{x}&=&\mu\mathbf{1}+\boldsymbol{\lambda}_{RE}+\boldsymbol{\varepsilon}_{RE}
\quad \text{with} \quad \boldsymbol{\lambda}_{RE} \sim N(\mathbf{0}, \tau_{RE}^2\mathbf{I})
\quad \text{and} \quad \boldsymbol{\varepsilon}_{RE} \sim N(\mathbf{0},\mathbf{U}).
\end{eqnarray}
Based on the above presentation we get that the random vector $\mathbf{x}$ is multivariate normally distributed, i.e., $\mathbf{x}|\mu,\tau_{RE},M_{RE}\sim N(\mu\mathbf{1},\mathbf{U}+\tau_{RE}^2\mathbf{I})$. In the random effects model an adjustment for the heterogeneity is made to $\mathbf{U}$ by adding a common term $\tau_{RE}^2$ to every element in the diagonal of $\mathbf{U}$, where $\sigma_{RE}^2$ represents the unexplainable interlaboratory heterogeneity. Since only the diagonals of $\mathbf{U}$ are modified by $\tau_{RE}^2$, the variances $u_i^2$ are adjusted for by the heterogeneity while the covariances in $\mathbf{U}$ are left unaltered.

Under the assumption of the random effects model the likelihood function is expressed as
	\begin{equation*}
f(\mathbf{x}|\mu, \tau_{RE}, M_{RE})=\frac{\big(\mathsf{det}(\mathbf{U}+\tau_{RE}^2\mathbf{I})\big)^{-\frac{1}{2}}}{(2\pi)^{n/2}}\mathsf{exp}\bigg(-\frac{1}{2}(\mathbf{x}-\mu\mathbf{1})^T
\big(\mathbf{U}+\tau_{RE}^2\mathbf{I}\big)^{-1}(\mathbf{x}-\mu\mathbf{1})\bigg) \ .
	\end{equation*}
The Berger \& Bernado reference prior has been derived in \citet{BodnarLinkElster2015} and it is given by
\begin{equation}\label{prior:REM}
\pi^N(\tau_{RE})=\pi^N(\mu,\tau_{RE})\propto \sqrt{\tau_{RE}^{2}\cdot tr((\mathbf{U}+\tau_{RE}^2\mathbf{I})^{-2})} \ .
\end{equation}
It has been proven in \citet[p.32]{BodnarLinkElster2015} that $\pi^N(\mu,\tau_{RE}|\mathbf{x})$ is proper for the random effects model and the Berger and Bernado reference prior, if $n>1$ . This leads to the conclusion that the minimal training sample size is equal to $m=2$. Finally, the marginal distribution of $\mathbf{x}$ and of the training sample $\mathbf{x}_{\ell}$ are presented in Theorem \ref{th2}.

\begin{theorem} \label{th2}
Let $n>2$ and let $\mathbf{U}$ be positive definite. Then under the random effects model \eqref{model:REM} and reference prior \eqref{prior:REM},
\begin{enumerate}[(i)]
\item the marginal distribution of the whole sample is given by
\begin{equation}\label{m:REM}
	m(\mathbf{x}|M_{RE})=\int_0^{\infty}\frac{(2\pi)^{-\frac{n-1}{2}}\big(\mathsf{det}(\mathbf{U}+\tau_{RE}^2\mathbf{I})\big)^{-1/2}}
	{\sqrt{\mathbf{1}^T(\mathbf{U}+\tau_{RE}^2\mathbf{I})^{-1}\mathbf{1}}}\mathsf{exp}\Big(-\frac{1}{2}\mathbf{x}^T\mathbf{Q}(\tau_{RE}^2)
	\mathbf{x}\Big)\pi^N(\tau_{RE}) \ d\tau_{RE} \ ,
\end{equation}
where
\begin{equation}\label{Q:REM}
\mathbf{Q}(\tau_{RE}^2)=(\mathbf{U}+\tau_{RE}^2\mathbf{I})^{-1}
-\frac{(\mathbf{U}+\tau_{RE}^2\mathbf{I})^{-1}\mathbf{1}\mathbf{1}^T(\mathbf{U}+\tau_{RE}^2\mathbf{I})^{-1}}{\mathbf{1}^T(\mathbf{U}+\tau_{RE}^2\mathbf{I})^{-1}\mathbf{1}}.
\end{equation}
\item The size of a minimal training sample $\mathbf{x}_{\ell}$ is given by $m=2$, i.e., $\mathbf{x}_{\ell}=\{x_i,x_j\}$, where $x_i,x_j\in\mathbf{x}$ and $i\neq j$. Moreover, the marginal distribution of the training sample is given by
\begin{equation}\label{ml:REM}
	m(\mathbf{x}_{\ell}|M_{RE})=\int_0^{\infty}\frac{(2\pi)^{-\frac{1}{2}}\big(\mathsf{det}(\mathbf{U}_{\ell}+\tau_{RE}^2\mathbf{I})\big)^{-1/2}}
	{\sqrt{\mathbf{1}_2^T(\mathbf{U}_{\ell}+\tau_{RE}^2\mathbf{I})^{-1}\mathbf{1}_2}}\mathsf{exp}\Big(-\frac{1}{2}\mathbf{x}_{\ell}^T\mathbf{Q}_{\ell}(\tau_{RE}^2)
	\mathbf{x}_{\ell}\Big)\pi^N(\tau_{RE}) \ d\tau_{RE} \ ,
\end{equation}
where $\mathbf{U}_{\ell}$ is obtained from $\mathbf{U}$ by taking its elements lying on the intersections of the $\ell$ rows and $\ell$ columns, and
\begin{equation}\label{Ql:REM} \mathbf{Q}_{\ell}(\tau_{RE}^2)=(\mathbf{U}_{\ell}+\tau_{RE}^2\mathbf{I})^{-1}-\frac{(\mathbf{U}_{\ell}+\tau_{RE}^2\mathbf{I})^{-1}\mathbf{1}_2\mathbf{1}_2^T(\mathbf{U}_{\ell}+\tau_{RE}^2\mathbf{I})^{-1}}{\mathbf{1}_2^T(\mathbf{U}_{\ell}+\tau_{RE}^2\mathbf{I})^{-1}\mathbf{1}_2}
	\quad \text{with} \quad \mathbf{1}_2=(1,1)^T .
\end{equation}
\end{enumerate}		
\end{theorem}

The proof of Theorem \ref{th2} is given in the appendix (see, Section \ref{sec:app}). Both the marginal distributions of the whole sample and of the training sample are present as a one-dimensional integral which cannot be derived analytically. On the other side, they can be computed with a high precision by numerical integration, for example, by using the Simpson rule (see, \citet{givens2012computational}).

Since the range of $\tau_{RE}$ is from $0$ to $+\infty$, we make a transformation under the integrals in \eqref{m:REM} and \eqref{ml:REM} following \citet{bodnar2020bayesian}, defined by $\tau_{RE}=\mathsf{tan}(\omega)$ with the Jacobian equal to $|\frac{d}{d\omega}\mathsf{tan}(\omega)|=1/\mathsf{cos}(\omega)^2$. It is a one-to-one map of the interval $\omega\in(0,\pi/2)$ into $\tau_{RE}\in(0,\infty)$. This allows us to compute the integral over the bounded range $\omega\in(0,\pi/2)$ instead over the unbounded range $\tau_{RE}\in(0,\infty)$. In the case of \eqref{m:REM} we then get
\begin{eqnarray*}
m(\mathbf{x}|M_{RE})&=&\int_0^{\pi/2}\frac{(2\pi)^{-\frac{n-1}{2}}\big(\mathsf{det}(\mathbf{U}+\mathsf{tan}^2(\omega)\mathbf{I})\big)^{-1/2}}{\sqrt{\mathbf{1}^T
(\mathbf{U}+\mathsf{tan}^2(\omega)\mathbf{I})^{-1}\mathbf{1}}}
\frac{\mathsf{exp}\Big(-\frac{1}{2}\mathbf{x}^T\mathbf{Q}(\mathsf{tan}^2(\omega))\mathbf{x}\Big)}{\mathsf{tan}(\omega)\mathsf{cos}^2(\omega)} \ d\omega,
\end{eqnarray*}
with $\mathbf{Q}(\mathsf{tan}^2(\omega))$ defined in \eqref{Q:REM}. A similar transformation is also used in the computation of the integral in \eqref{ml:REM}.

\subsection{Location-scale model versus random effects model}\label{sec:IBF-LSM-REM}

The difference between the location-scale model and the random effects model is how they take into account for the heterogeneity which arise when the results of several studies are combined together. The choice stands between letting the heterogeneity be modeled as an multiplicative correction factor (location-scale model) or as an additive correction term (random effects model). In order to make a selection between the location-scale model and the random effects model we opt for the application of the IBF presented in Section \ref{sec:IBF} with the marginal distributions of the whole sample and of the training sample as derived in Theorems \ref{th1} and \ref{th2}.

The IBF for comparing the random effects model ($M_{RE}$) to the location-scale model ($M_{LS}$) and the training sample $\mathbf{x}_{\ell}$ is given by
\begin{equation}\label{IBF_RE_LS}
B_{M_{RE}M_{LS}}^I(\mathbf{x}_{(\ell)}|\mathbf{x}_{\ell})= B_{M_{RE}M_{LS}}^N(\mathbf{x})\times B_{M_{LS}M_{RE}}^N(\mathbf{x}_{\ell})
	=\frac{m(\mathbf{x}|M_{RE})}{m(\mathbf{x}|M_{LS})}\times \frac{m(\mathbf{x}_{\ell}|M_{LS})}{m(\mathbf{x}_{\ell}|M_{RE})},
\end{equation}
where $m(\mathbf{x}|M_{RE})$, $m(\mathbf{x}|M_{LS})$, $m(\mathbf{x}_{\ell}|M_{LS})$, and $m(\mathbf{x}_{\ell}|M_{RE})$ are given in \eqref{m:REM}, \eqref{m:LSM}, \eqref{ml:LSM}, and \eqref{ml:REM}, respectively. The value of $B_{M_{RE}M_{LS}}^I(\mathbf{x}_{(\ell)}|\mathbf{x}_{\ell})$ larger than one indicates that the random effects model is preferable, while the inequality $B_{M_{RE}M_{LS}}^I(\mathbf{x}_{(\ell)}|\mathbf{x}_{\ell})<1$ suggests the application of the location-scale model.

Both the location-scale model and the random effects model have been shown to have a minimal training sample size of $m=2$ (see, Theorems \ref{th1} and \ref{th2}). With a total sample size of $n$, we have a number of $n(n-1)/2$ combinations to select two observations out of $n$. Let $\mathcal{X}_{\ell} = (\mathbf{x}_{\ell}^1, \mathbf{x}_{\ell}^2,\cdots, \mathbf{x}_{\ell}^{l-1}, \mathbf{x}_{\ell}^l)$ be the set of training samples, there are $L=n(n-1)/2$ number of elements in $\mathcal{X}_{\ell}$. The IBF \eqref{IBF_RE_LS} is computed for every such possible selection and is aggregated over the possible training samples by computing the average or the median. The reason is that the mean and median of the IBF over $\mathcal{X}_{\ell}$ has increased stability in comparison with the IBF of an arbitrary $\mathbf{x}_{\ell}^i$. If the sample size $n$ is either too small (too slight stability improvement) or too large (too long computing time), it is then possible to use the expected mean discussed in \citet[p.113-114]{berger1996intrinsic}.

\citet[p.149]{berger2001objective} noted that it is recommendable to put the more ''complex'' model in the numerator of \eqref{IBF_RE_LS}, when the average of the IBF is computed over the training samples. One of the reasons is that some large values of the IBF computed for some training samples might dominate in the resulting value of the average IBF, although most of the IBFs are close to zero. In order to symmetrize the impact of large and small values of the IBF computed for each of the training samples, the logarithmic transformation is used before calculating the average IBF and the median IBF. It leads to the following formulas of the average IBF and the median IBF given by
\begin{equation}\label{aIBF_RE_LS}
aB_{M_{RE}M_{LS}}^I=\frac{1}{L}\sum_{\ell} \log\left(B_{M_{RE}M_{LS}}^I(\mathbf{x}_{(\ell)}|\mathbf{x}_{\ell})\right)
\end{equation}
and
\begin{equation}\label{mIBF_RE_LS}
mB_{M_{RE}M_{LS}}^I=\text{median}\left(\log\left(B_{M_{RE}M_{LS}}^I(\mathbf{x}_{(\ell)}|\mathbf{x}_{\ell})\right)\right),
\end{equation}
respectively.

Besides that, we also computed the empirical probability of the random effect model to be preferable which is expressed as
\begin{equation}\label{epIBF_RE_LS}
epB_{M_{RE}M_{LS}}^I=\frac{1}{L}\sum_{\ell} \mathds{1}_{(0,\infty)}\left(\log\left(B_{M_{RE}M_{LS}}^I(\mathbf{x}_{(\ell)}|\mathbf{x}_{\ell})\right)\right),
\end{equation}
where $\mathds{1}_{\mathcal{A}}(\cdot)$ is the indicator function of set $\mathcal{A}$.

Based on the computed values in \eqref{aIBF_RE_LS}, \eqref{mIBF_RE_LS}, and \eqref{epIBF_RE_LS}, one prefer the random effects model to the location-scale model if $aB_{M_{RE}M_{LS}}^I>0$, $mB_{M_{RE}M_{LS}}^I>0$, or $epB_{M_{RE}M_{LS}}^I>0.5$ depending on the selected criterium. Otherwise, the location-scale model is preferable.

\section{Simulation study}\label{sec:sim}
In this section we investigate the performance of the Bayesian model selection based on the intrinsic Bayes factors described in Section \ref{sec:IBF-LSM-REM} by drawing samples from the location-scale model and the random effects model, respectively. By doing so we will check whether the model selection based on the considered IBFs leads to the model from which the data were generated and study the impact of the sample size on the decision.

\begin{figure}[h!t]
\centering
\begin{tabular}{cc}
\includegraphics[width=8cm]{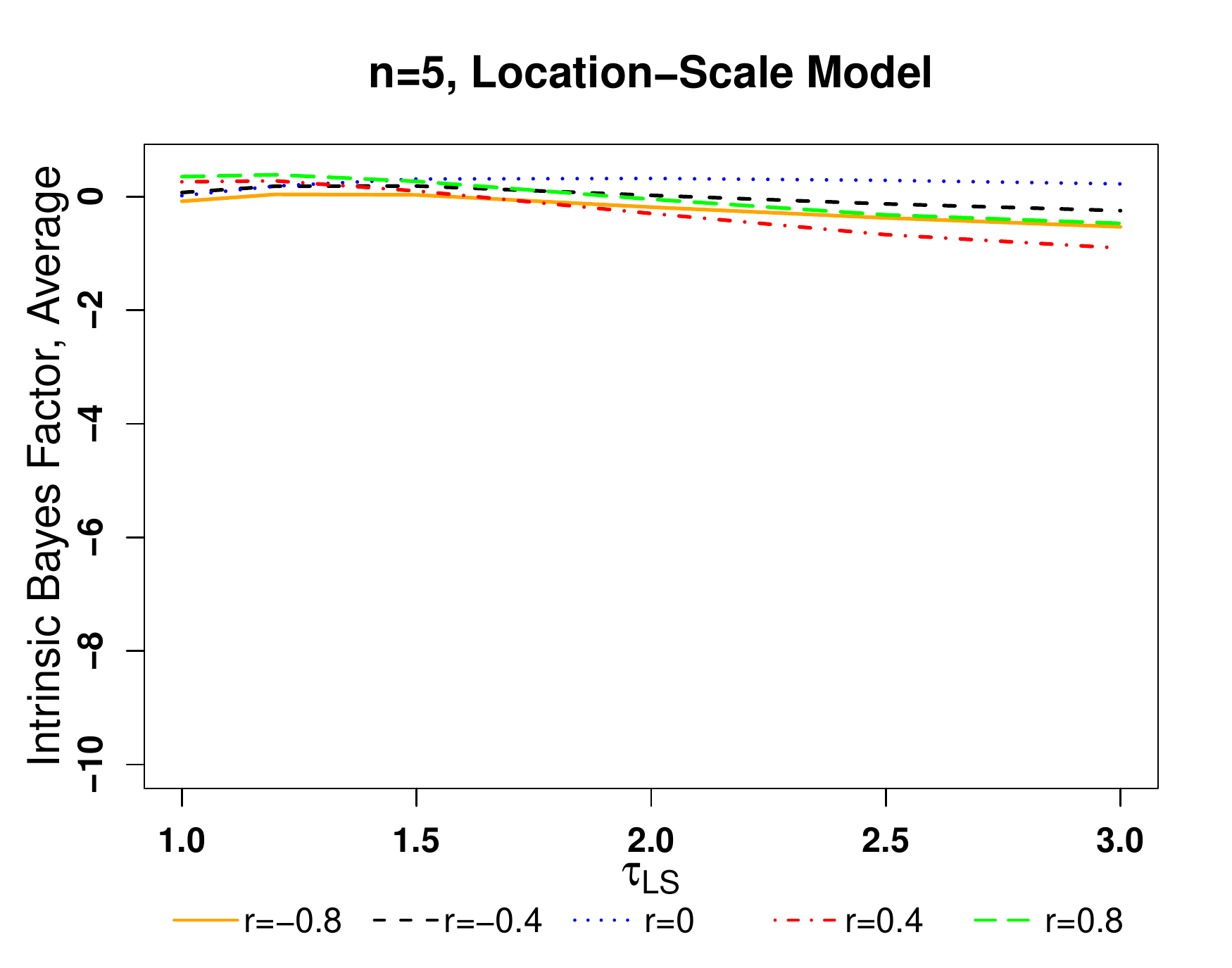}&\includegraphics[width=8cm]{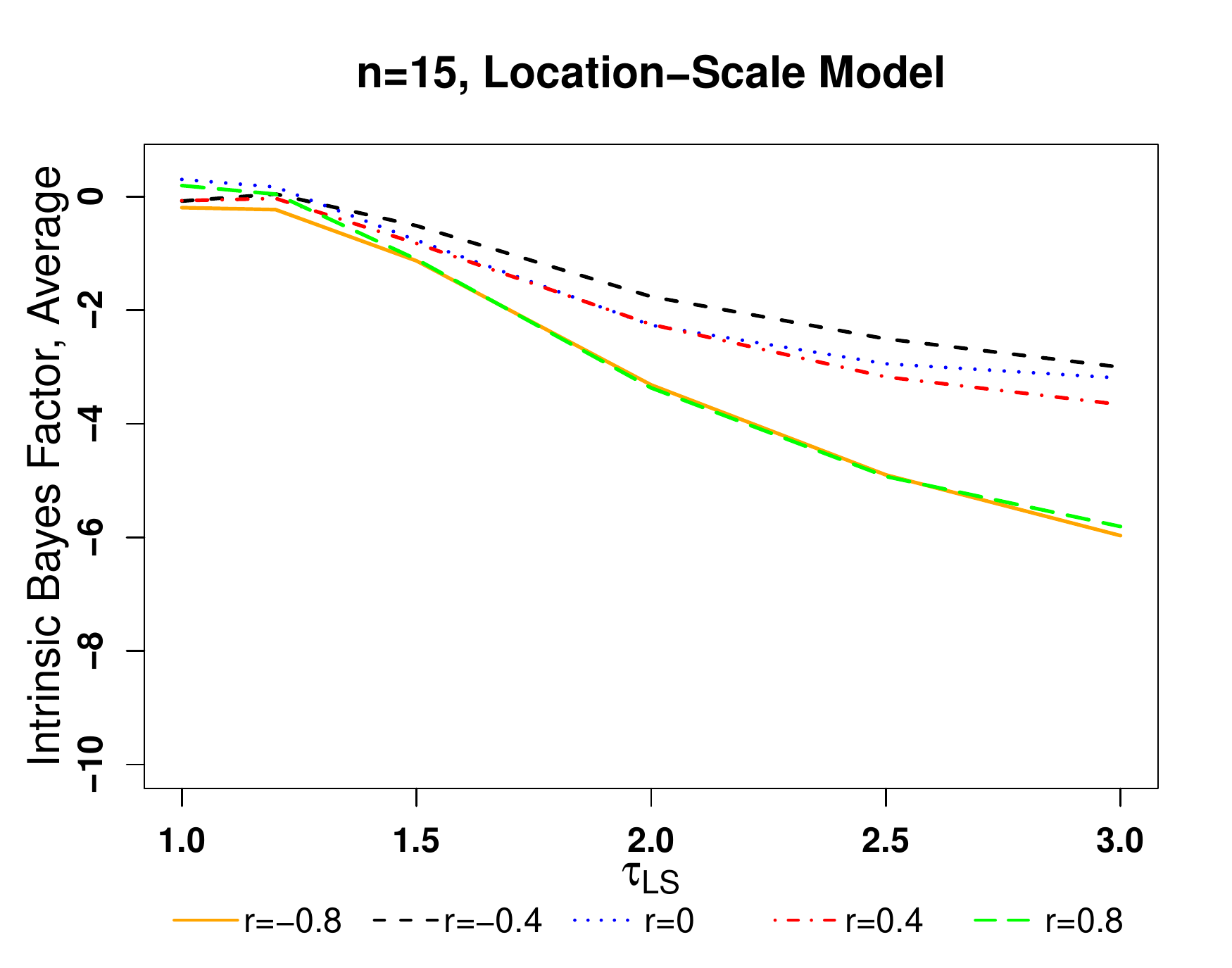}\\
\includegraphics[width=8cm]{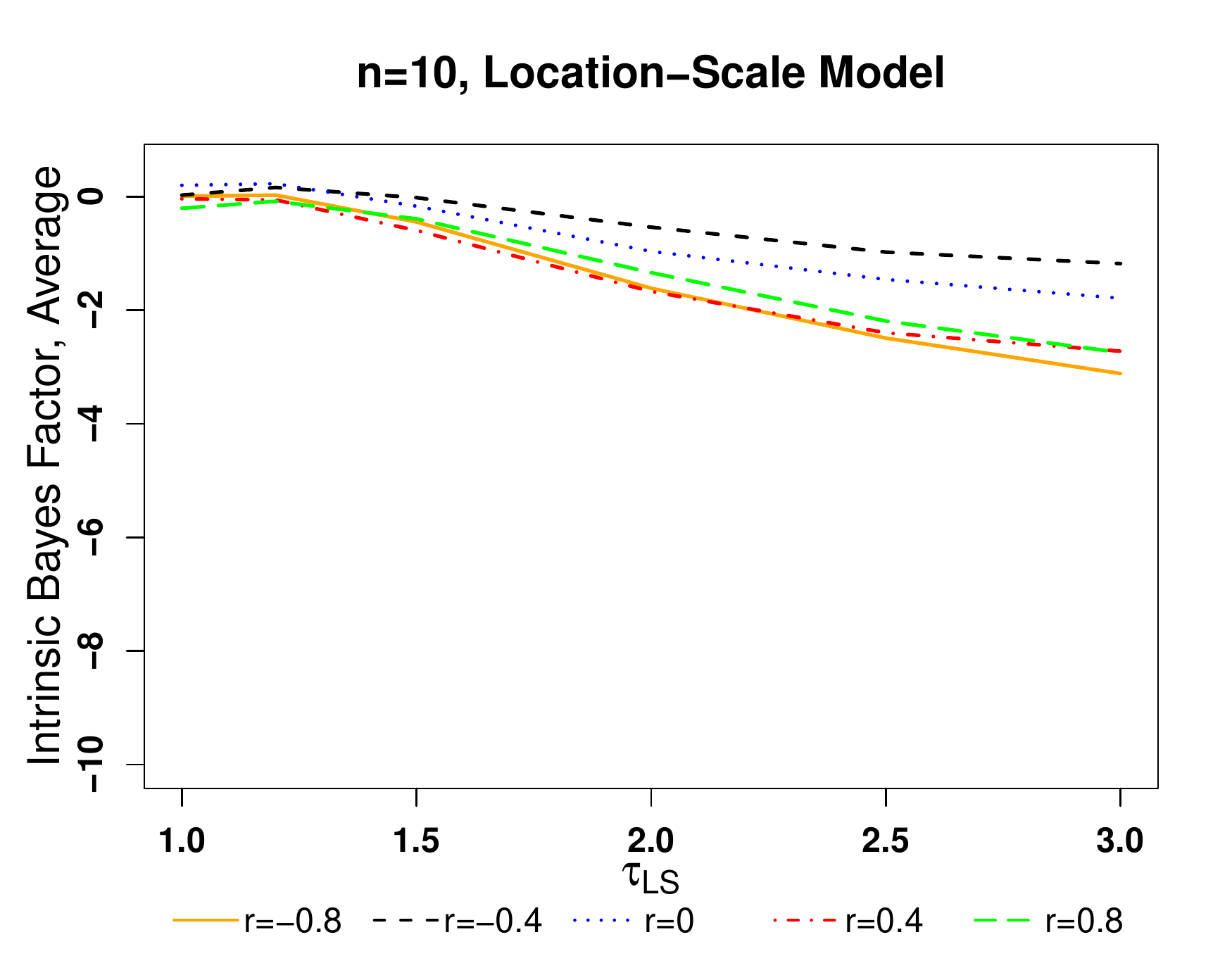}&\includegraphics[width=8cm]{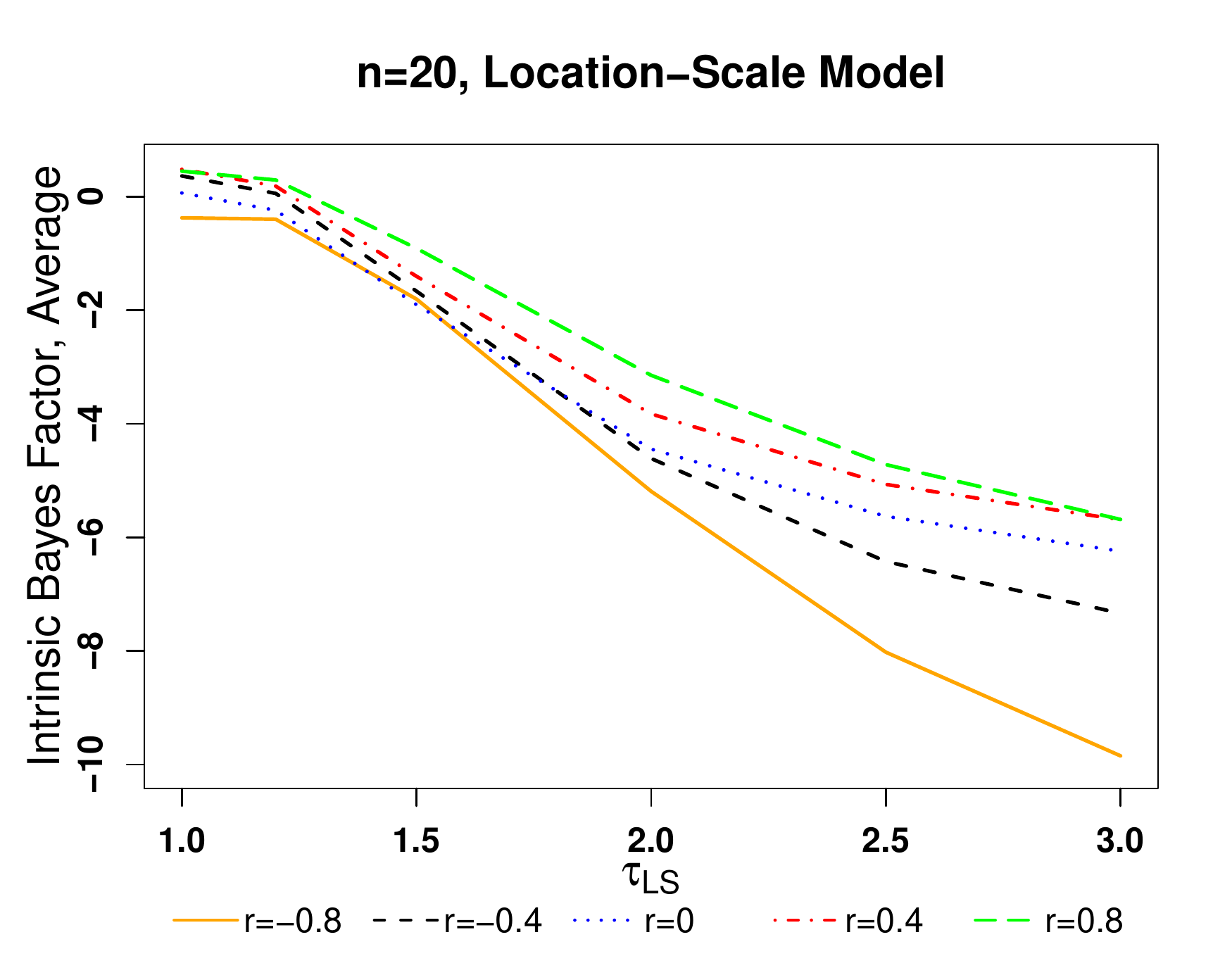}\\
\end{tabular}
 \caption{Average intrinsic Bayes factor \eqref{aIBF_RE_LS} for comparing the random effects model to the location-scale model as a function of $\tau_{LS}$ when the reference prior is employed. We set $n\in\{5,10,15,20\}$ and $r\in \{-0.8,-0.4,0,0.4,0.8\}$. The data were drawn from the location-scale model (see, Scenario 1).}
\label{fig:aIBF_LS}
 \end{figure}

\begin{figure}[h!t]
\centering
\begin{tabular}{cc}
\includegraphics[width=8cm]{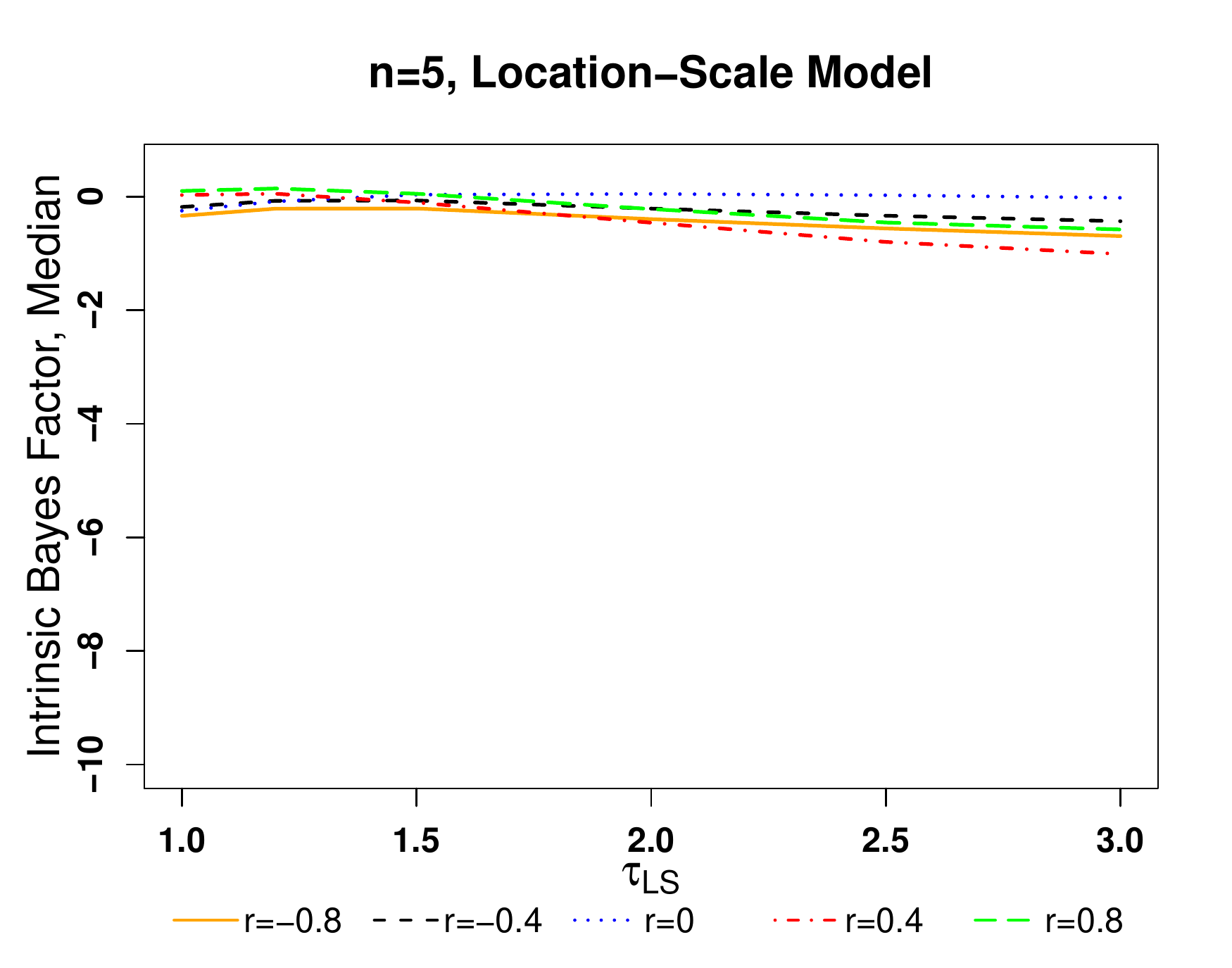}&\includegraphics[width=8cm]{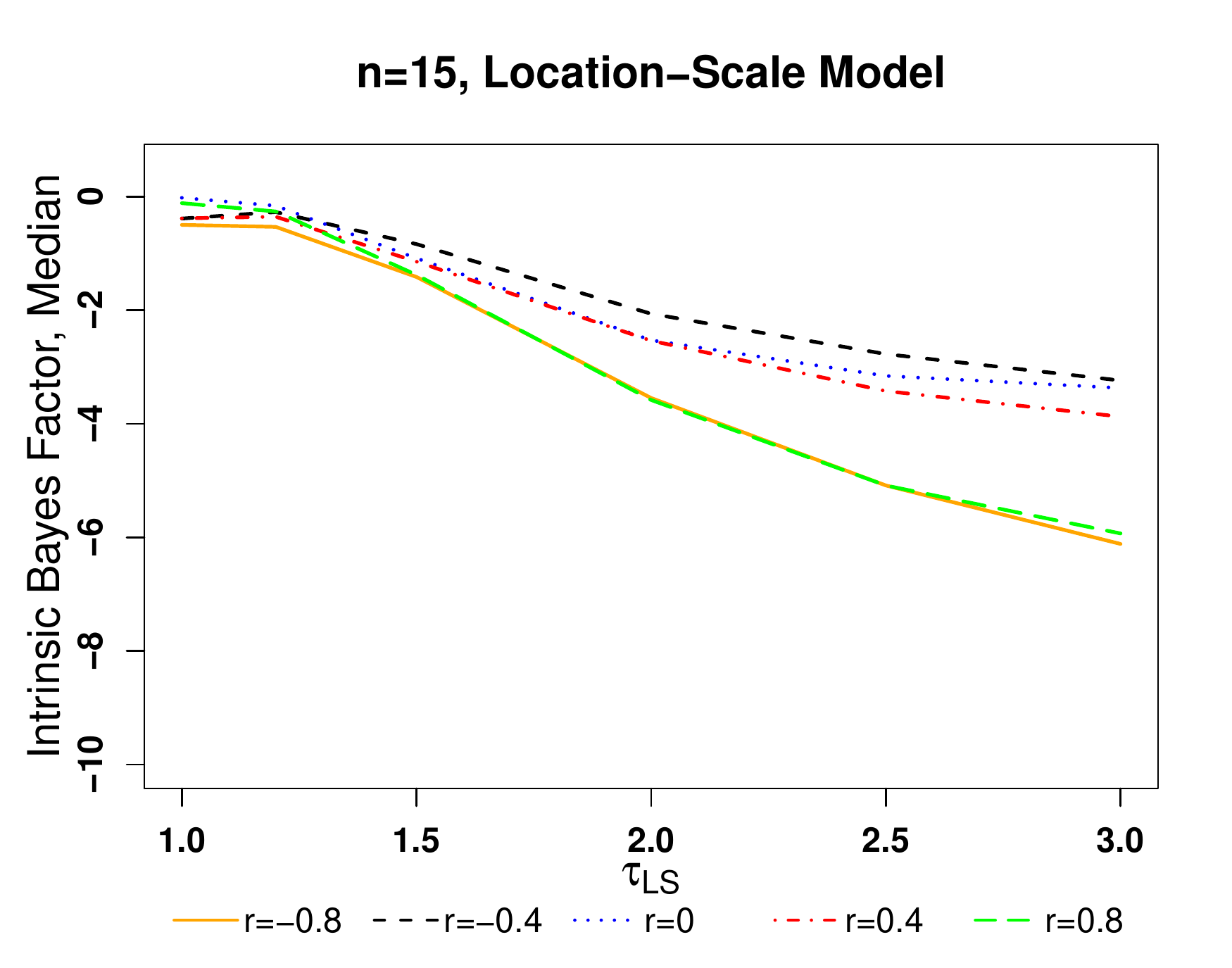}\\
\includegraphics[width=8cm]{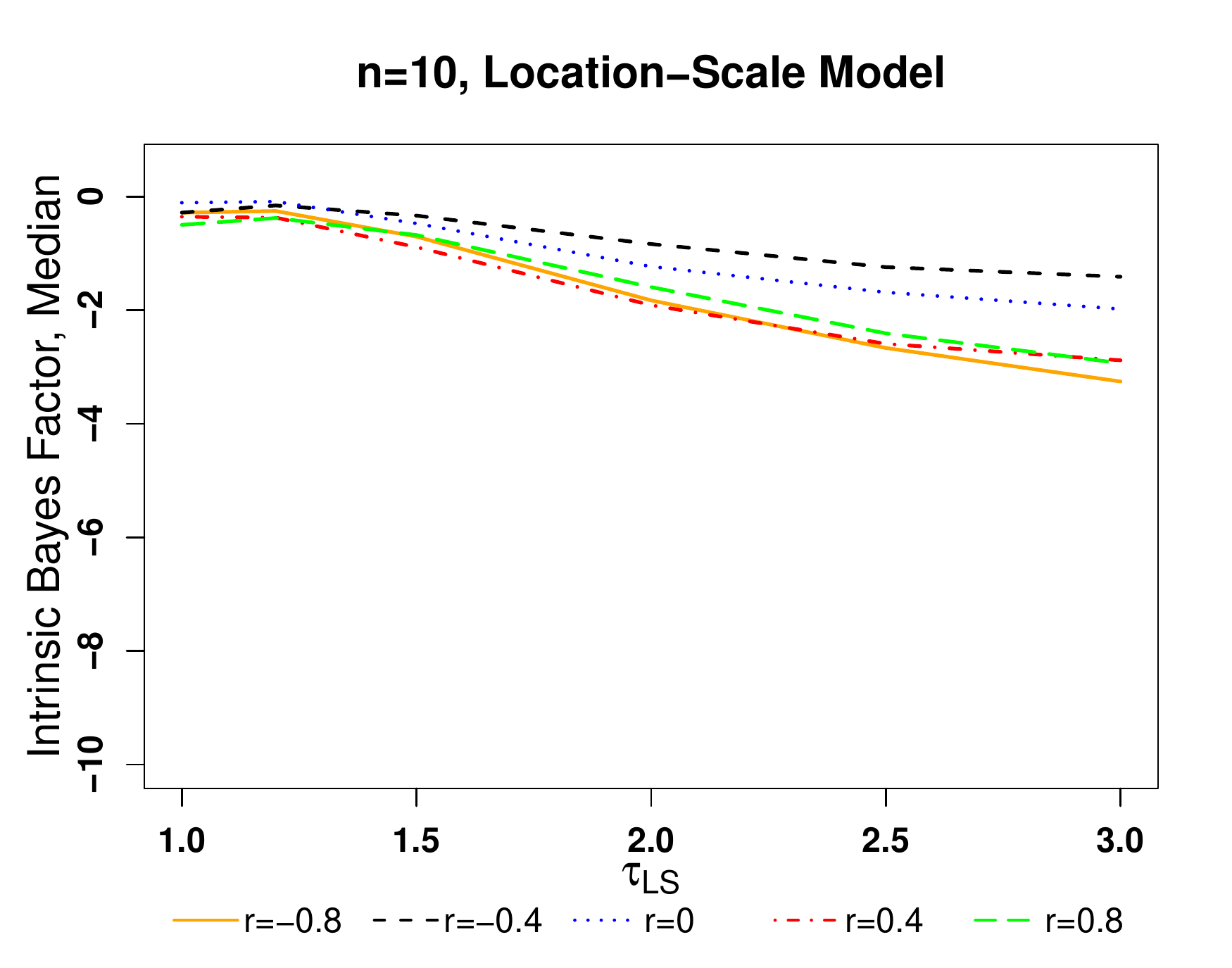}&\includegraphics[width=8cm]{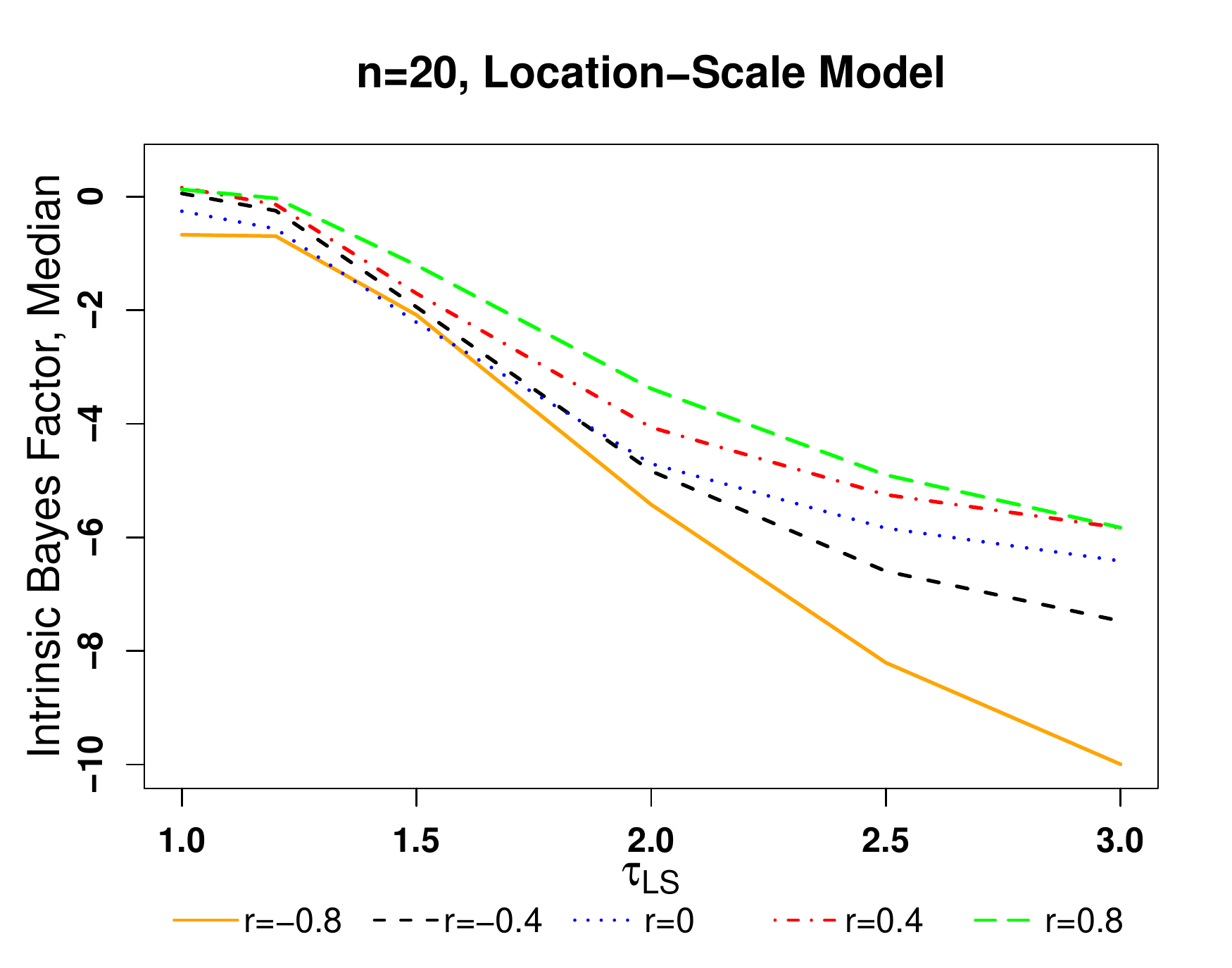}\\
\end{tabular}
 \caption{Median intrinsic Bayes factor \eqref{mIBF_RE_LS} for comparing the random effects model to the location-scale model as a function of $\tau_{LS}$ when the reference prior is employed. We set $n\in\{5,10,15,20\}$ and $r \in \{-0.8,-0.4,0,0.4,0.8\}$. The data were drawn from the location-scale model (see, Scenario 1).}
\label{fig:mIBF_LS}
 \end{figure}

\begin{figure}[h!t]
\centering
\begin{tabular}{cc}
\includegraphics[width=8cm]{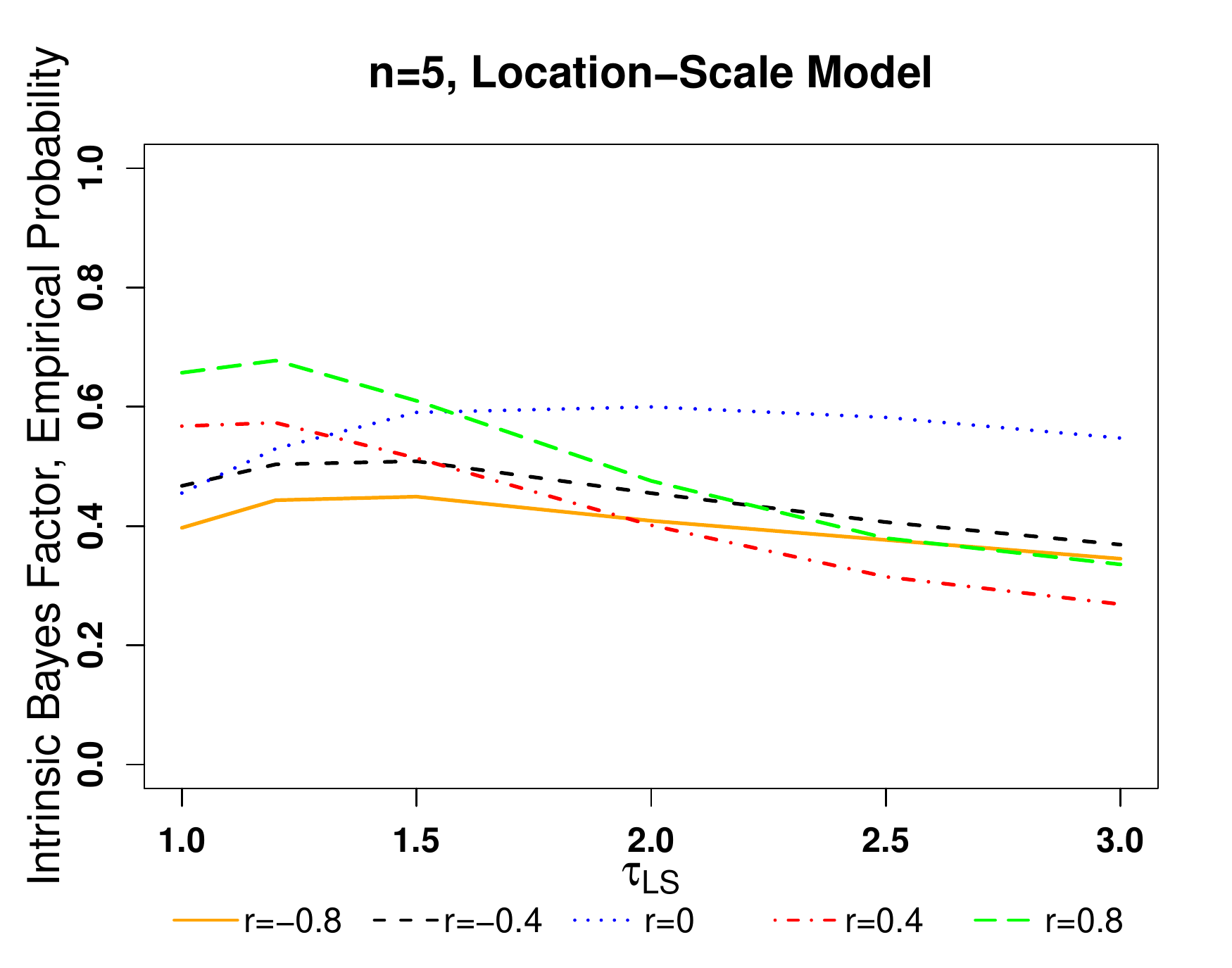}&\includegraphics[width=8cm]{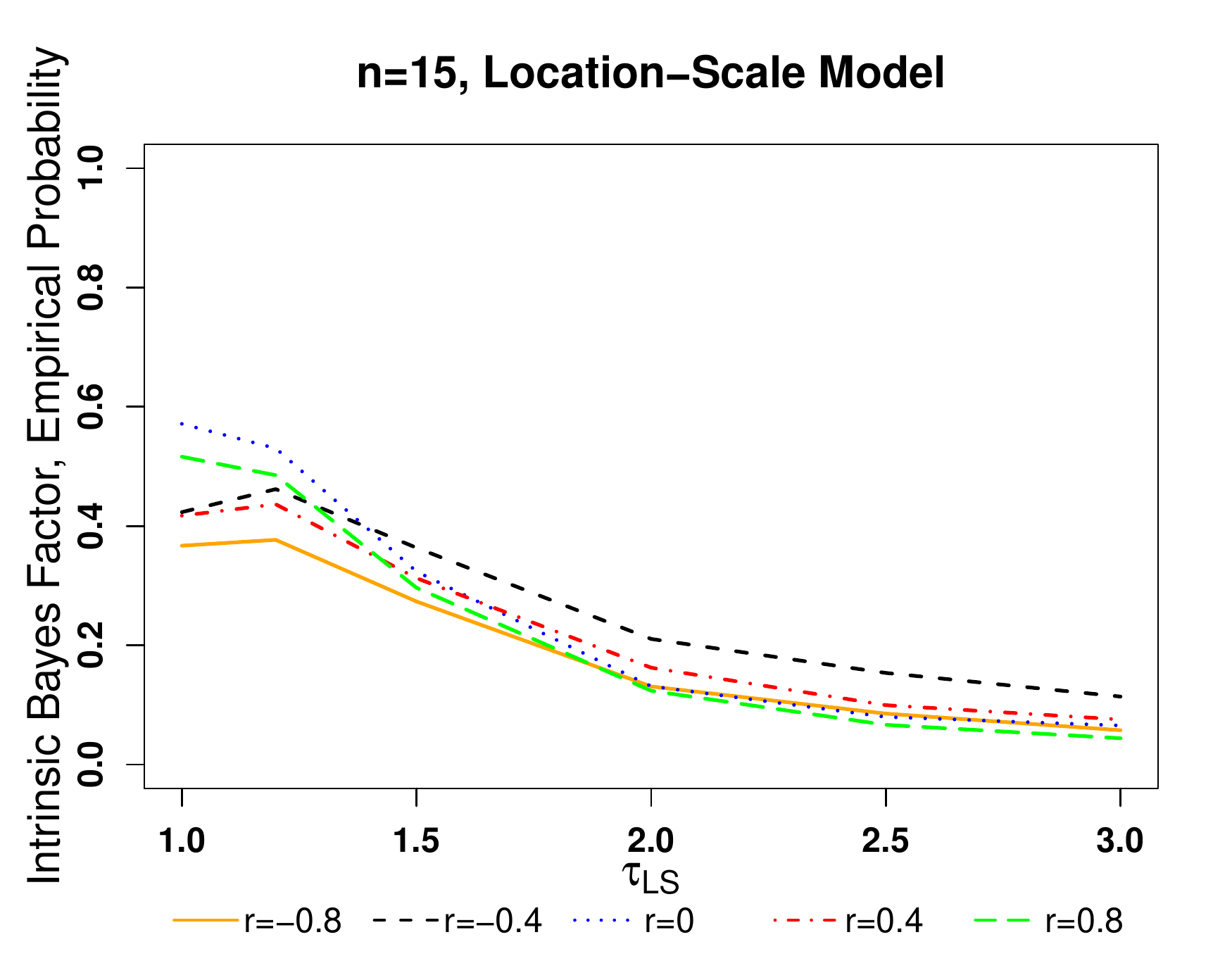}\\
\includegraphics[width=8cm]{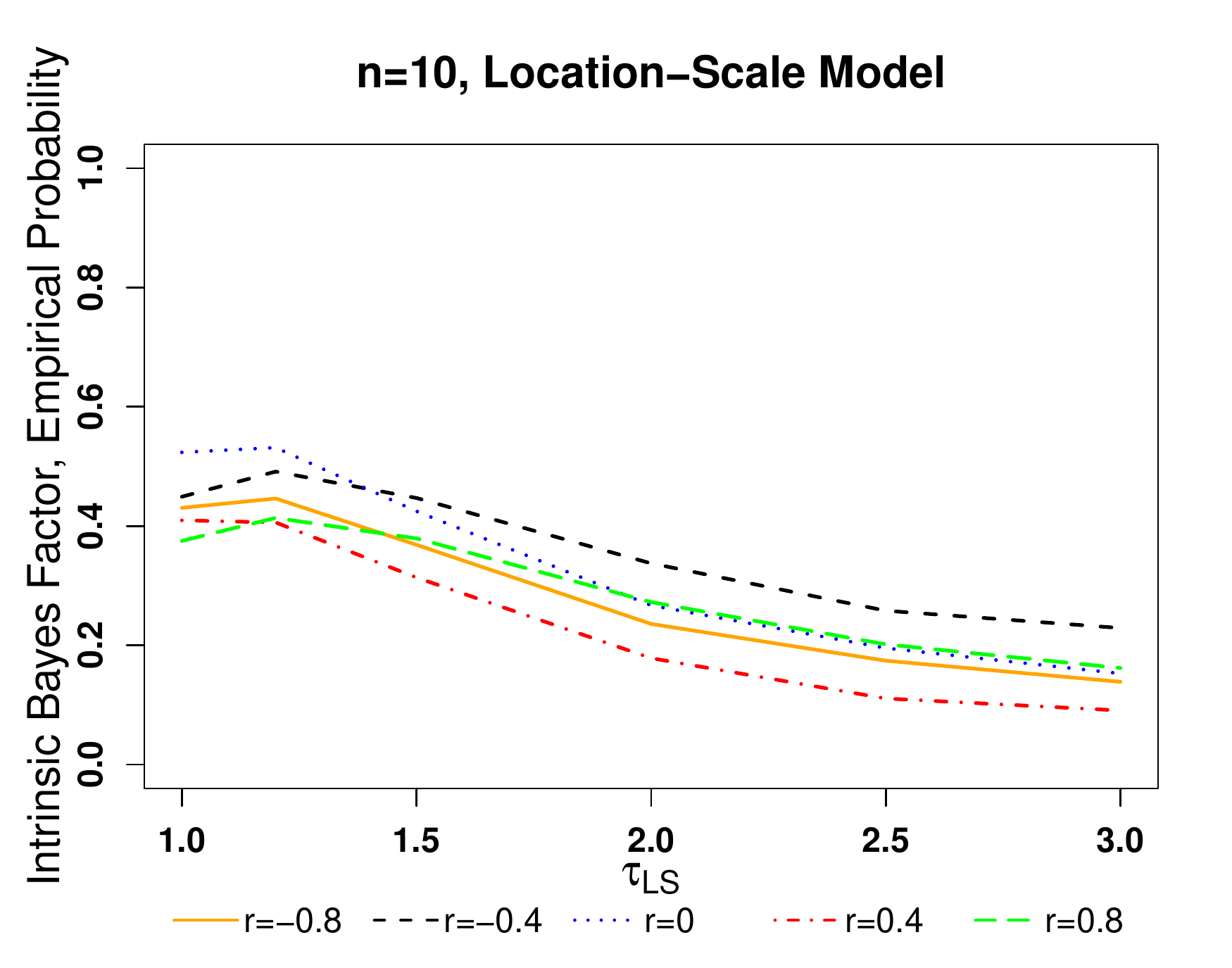}&\includegraphics[width=8cm]{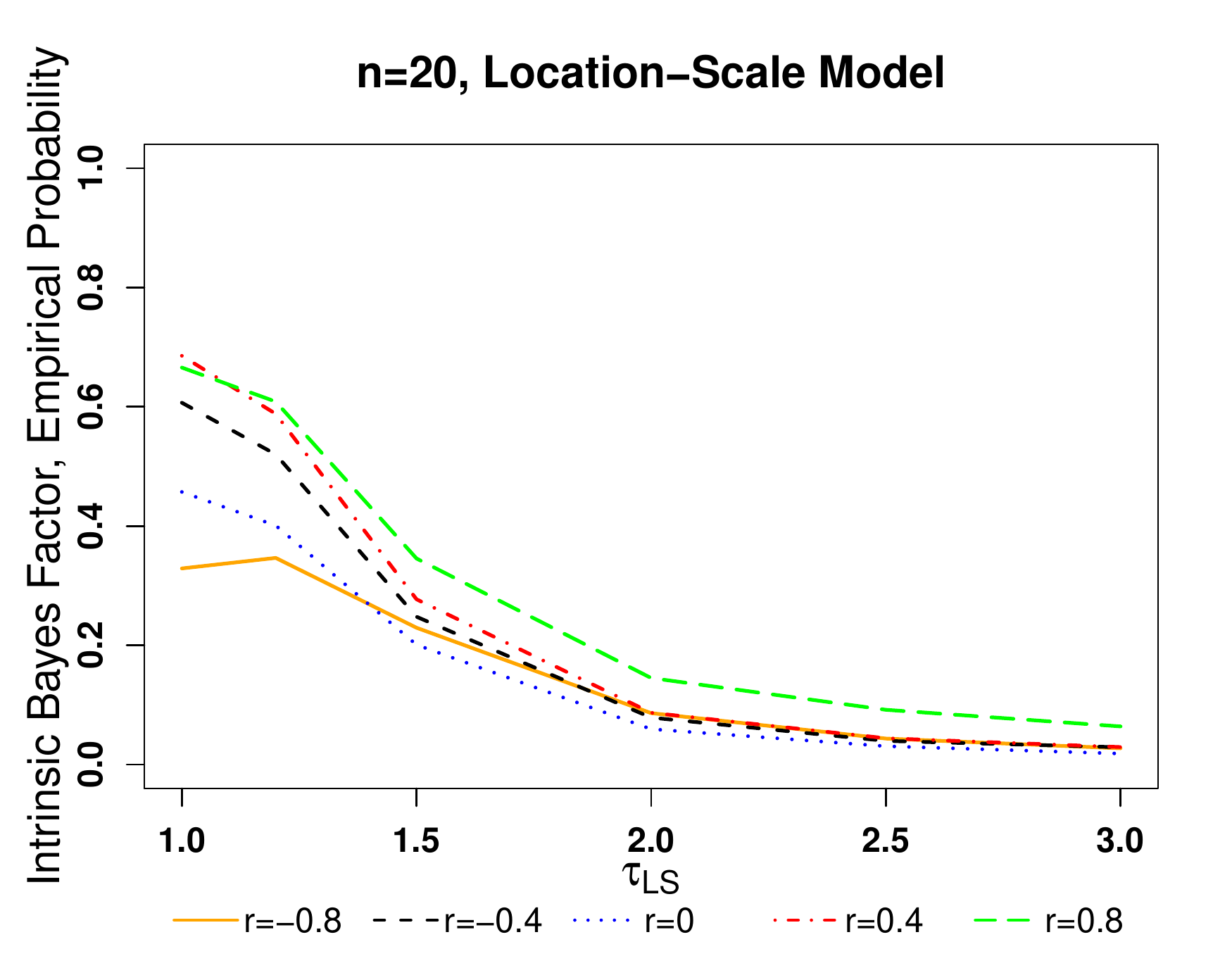}\\
\end{tabular}
 \caption{Empirical probability intrinsic Bayes factor \eqref{epIBF_RE_LS} for comparing the random effects model to the location-scale model as a function of $\tau_{LS}$ when the reference prior is employed. We set $n\in\{5,10,15,20\}$ and $r \in \{-0.8,-0.4,0,0.4,0.8\}$. The data were drawn from the location-scale model (see, Scenario 1).}
\label{fig:epIBF_LS}
 \end{figure}

\begin{figure}[h!t]
\centering
\begin{tabular}{cc}
\includegraphics[width=8cm]{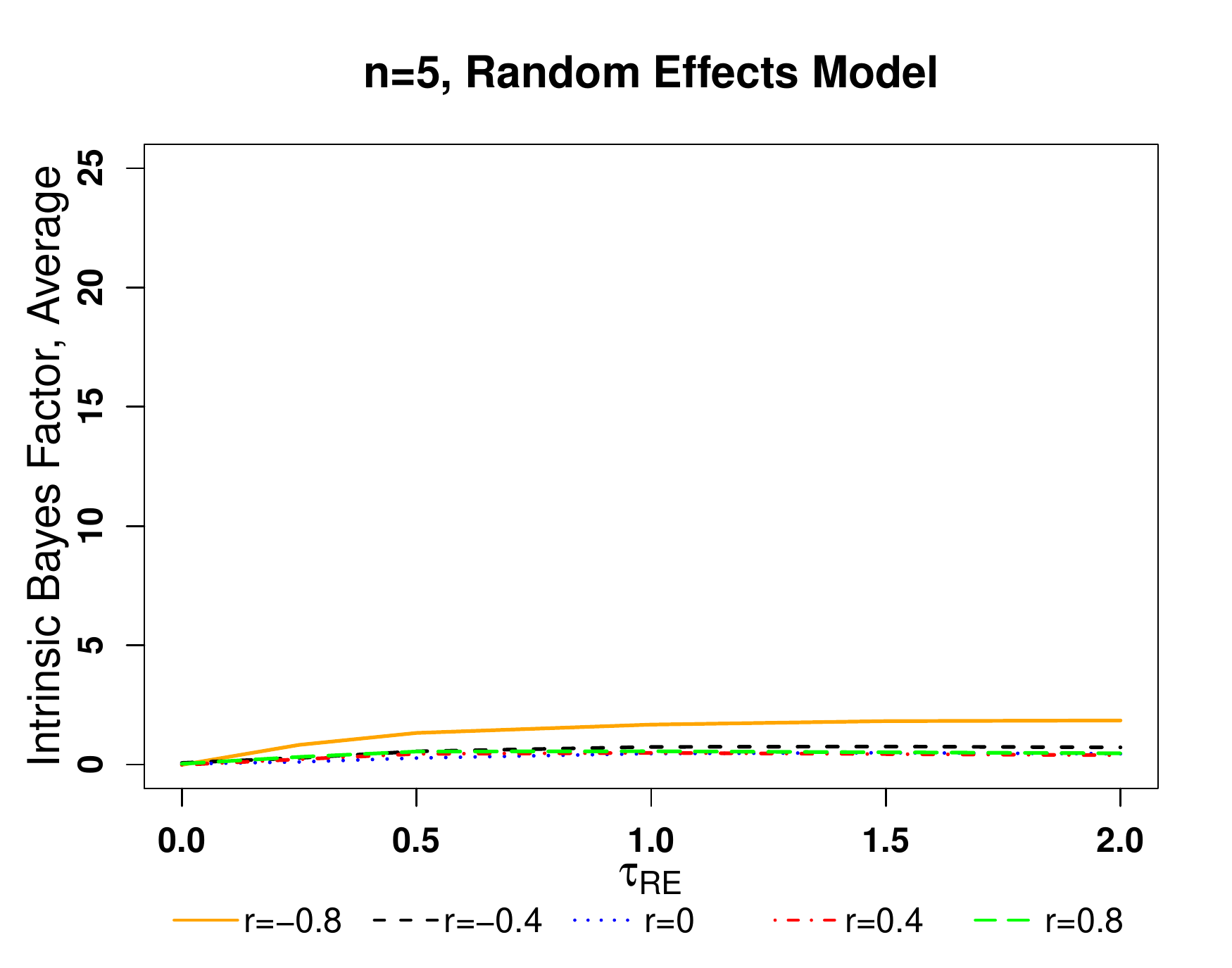}&\includegraphics[width=8cm]{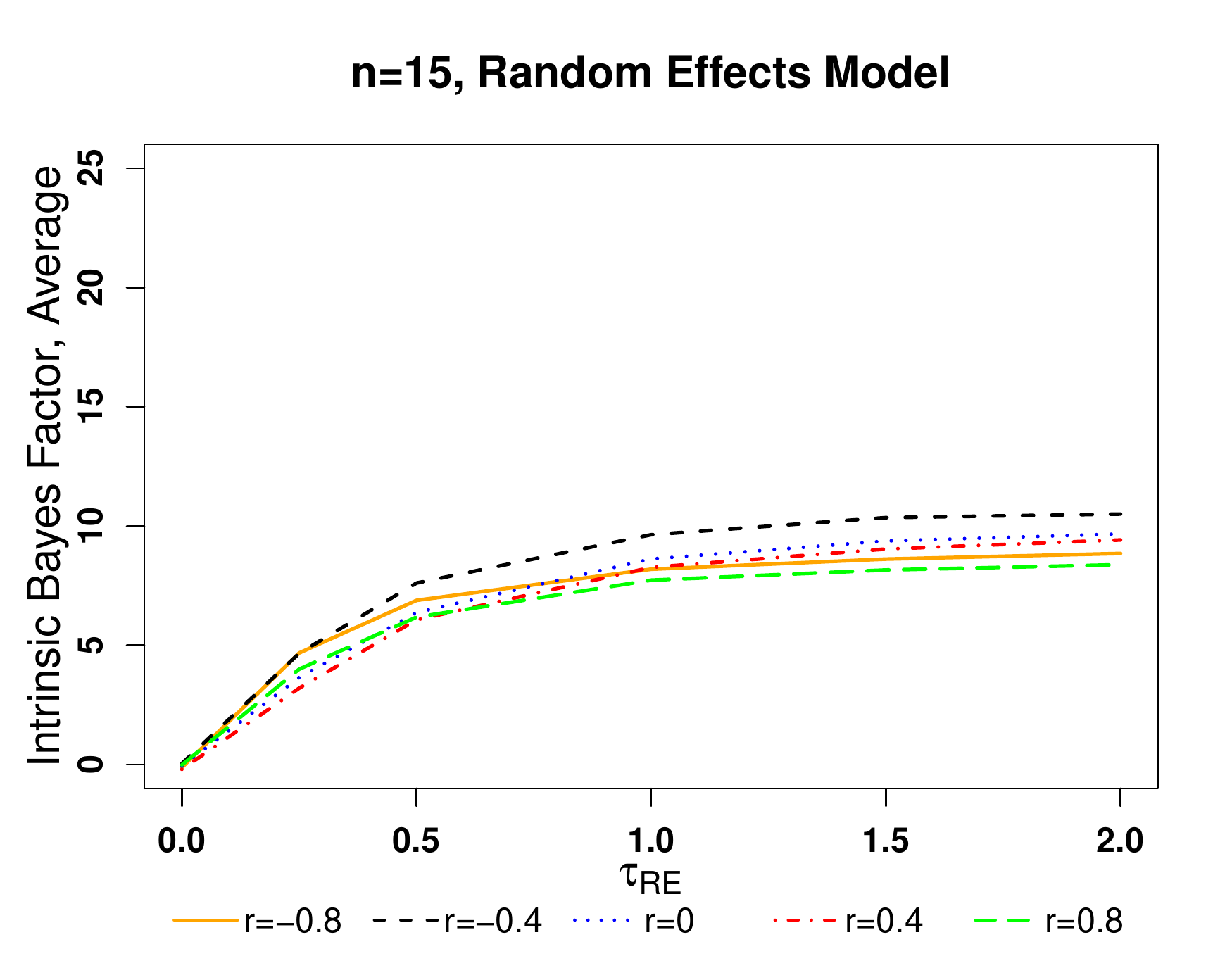}\\
\includegraphics[width=8cm]{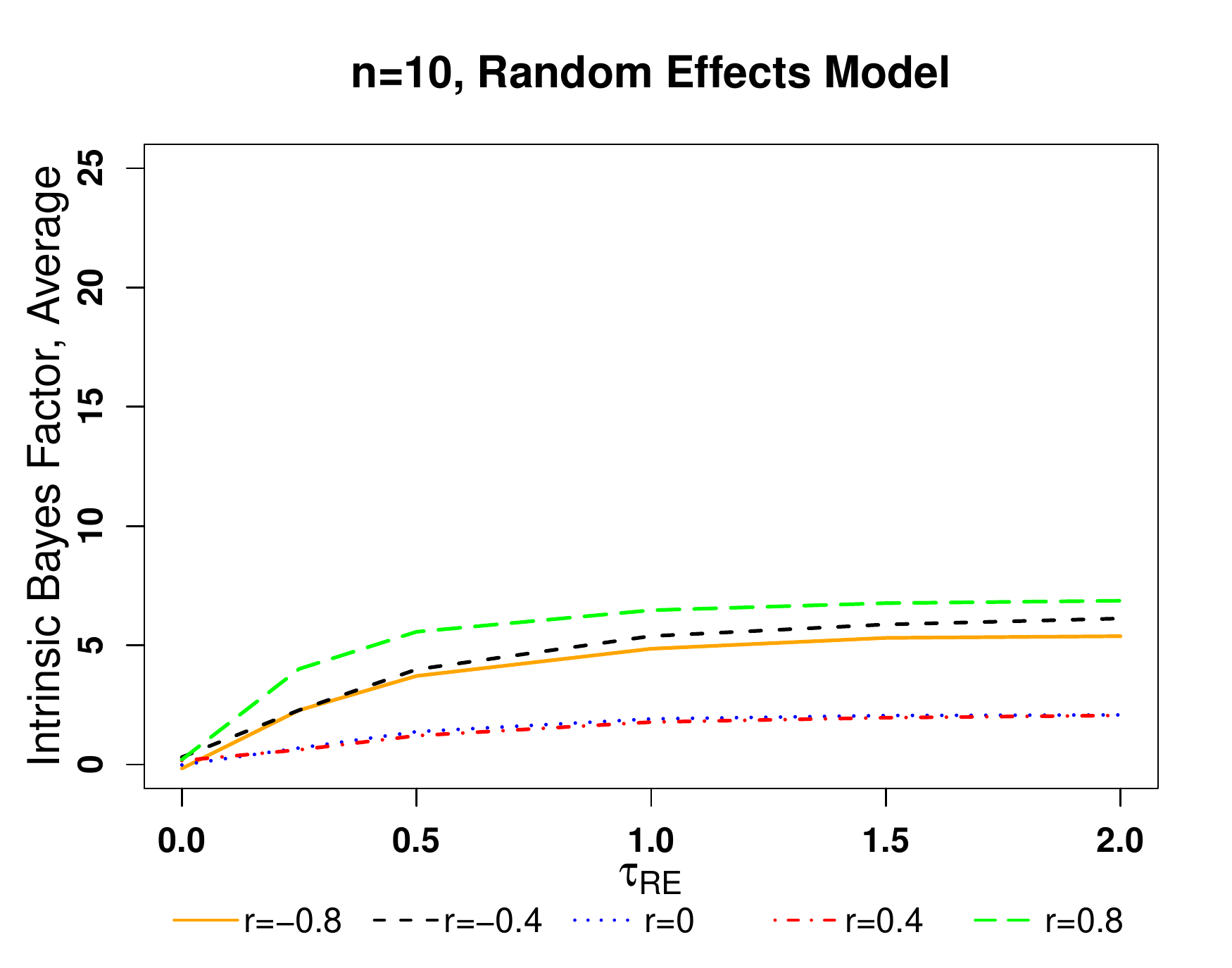}&\includegraphics[width=8cm]{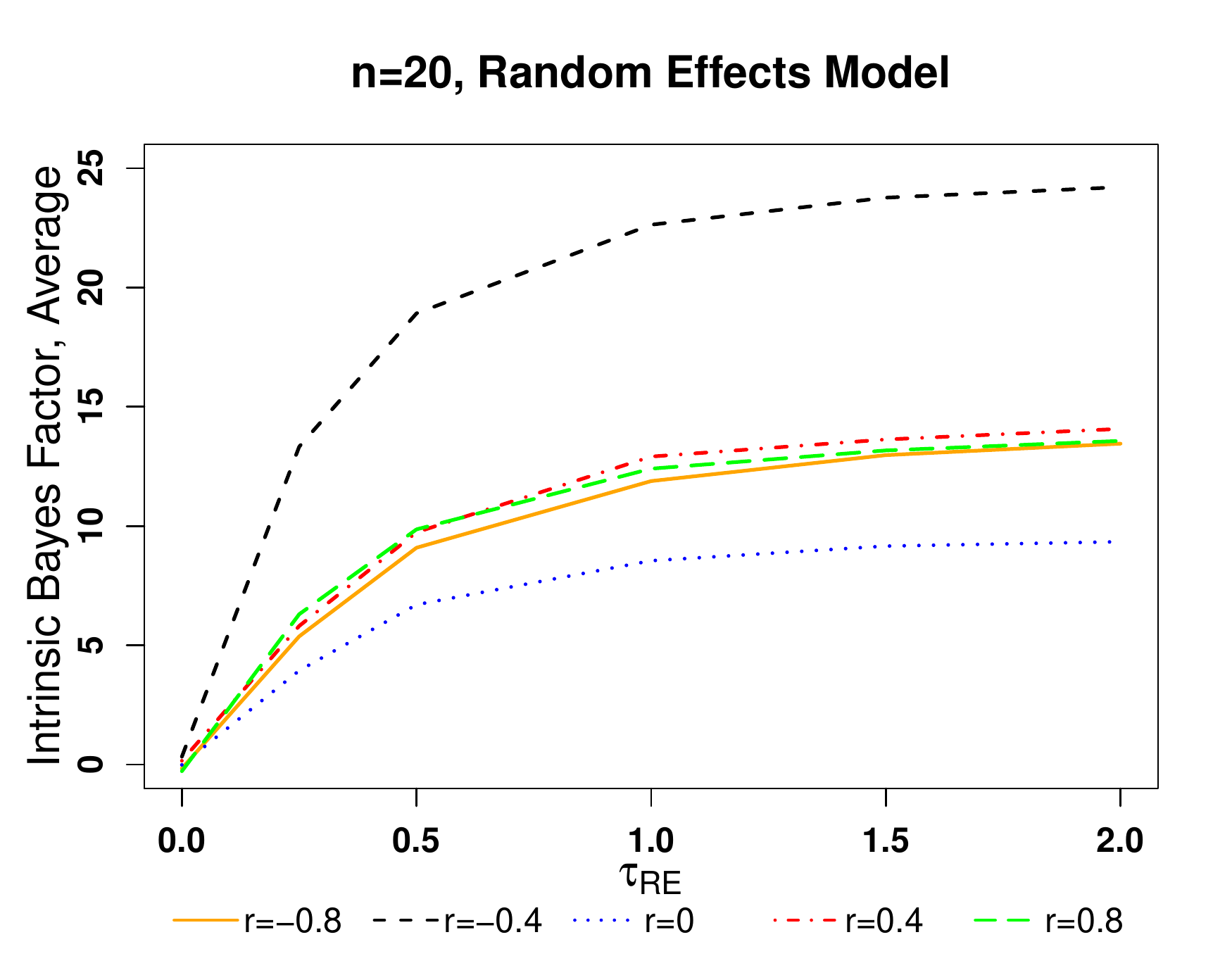}\\
\end{tabular}
 \caption{Average intrinsic Bayes factor \eqref{aIBF_RE_LS} for comparing the random effects model to the location-scale model as a function of $\tau_{RE}$ when the reference prior is employed. We set $n\in\{5,10,15,20\}$ and $r \in \{-0.8,-0.4,0,0.4,0.8\}$. The data were drawn from the random effects model (see, Scenario 2).}
\label{fig:aIBF_RE}
 \end{figure}

\begin{figure}[h!t]
\centering
\begin{tabular}{cc}
\includegraphics[width=8cm]{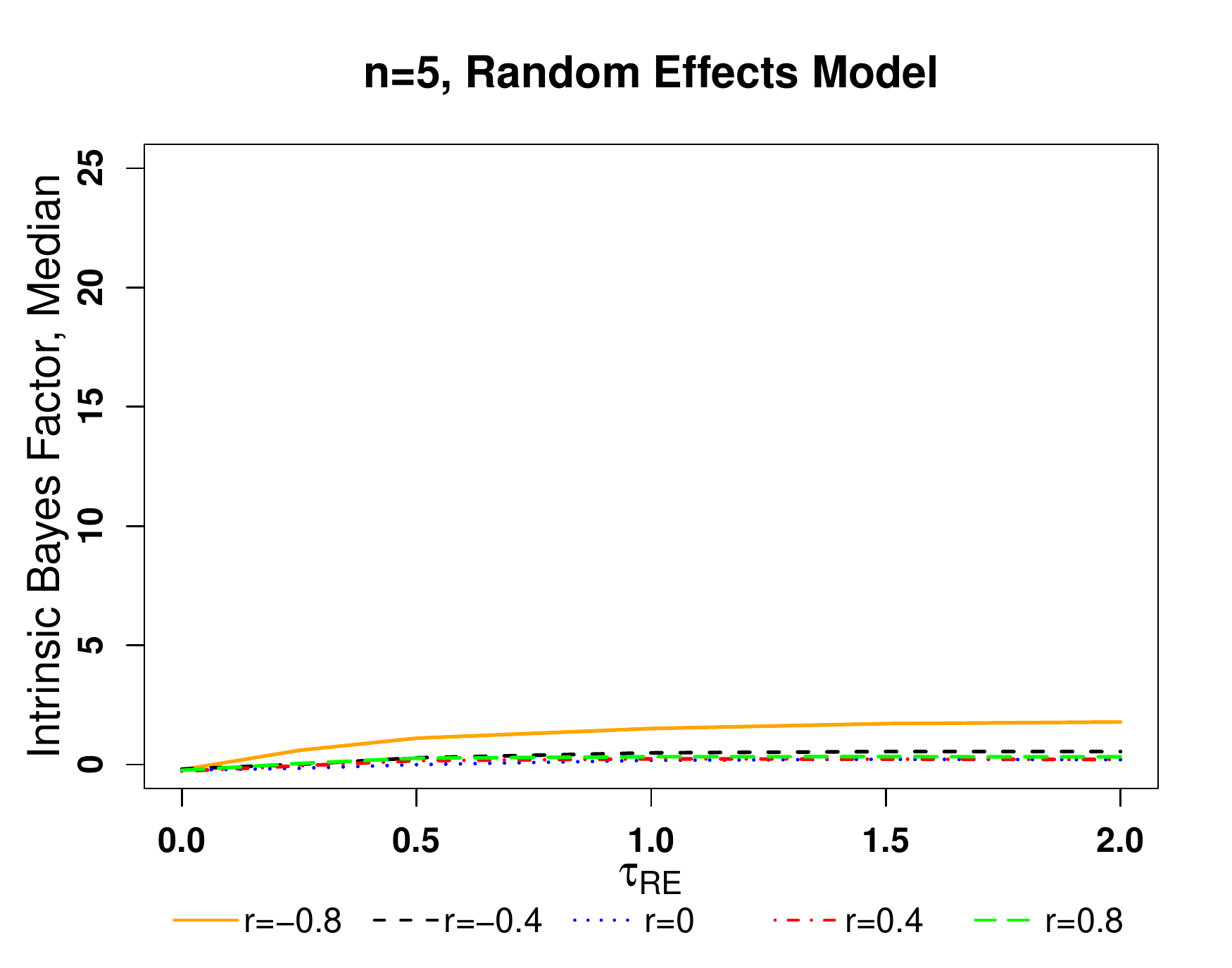}&\includegraphics[width=8cm]{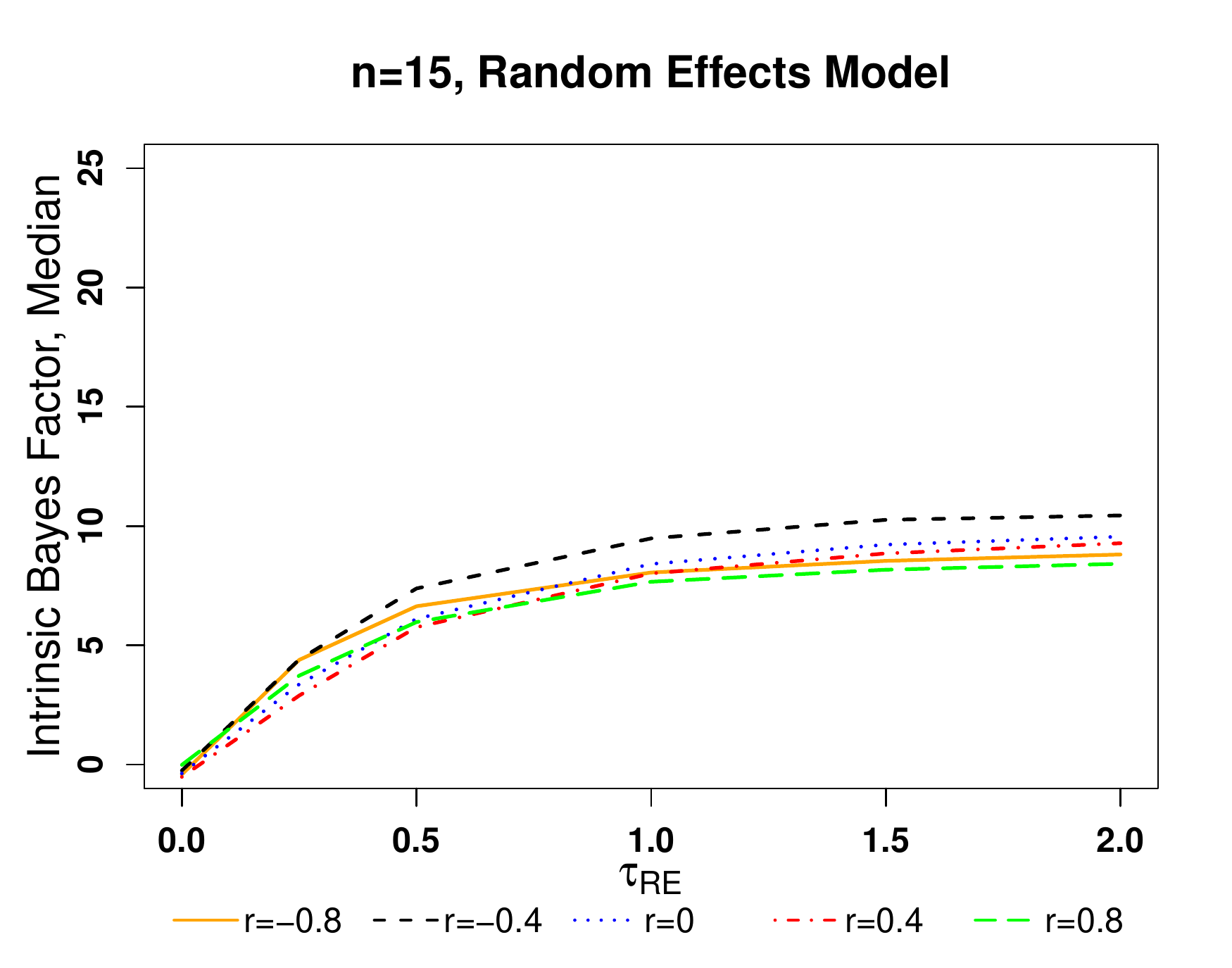}\\
\includegraphics[width=8cm]{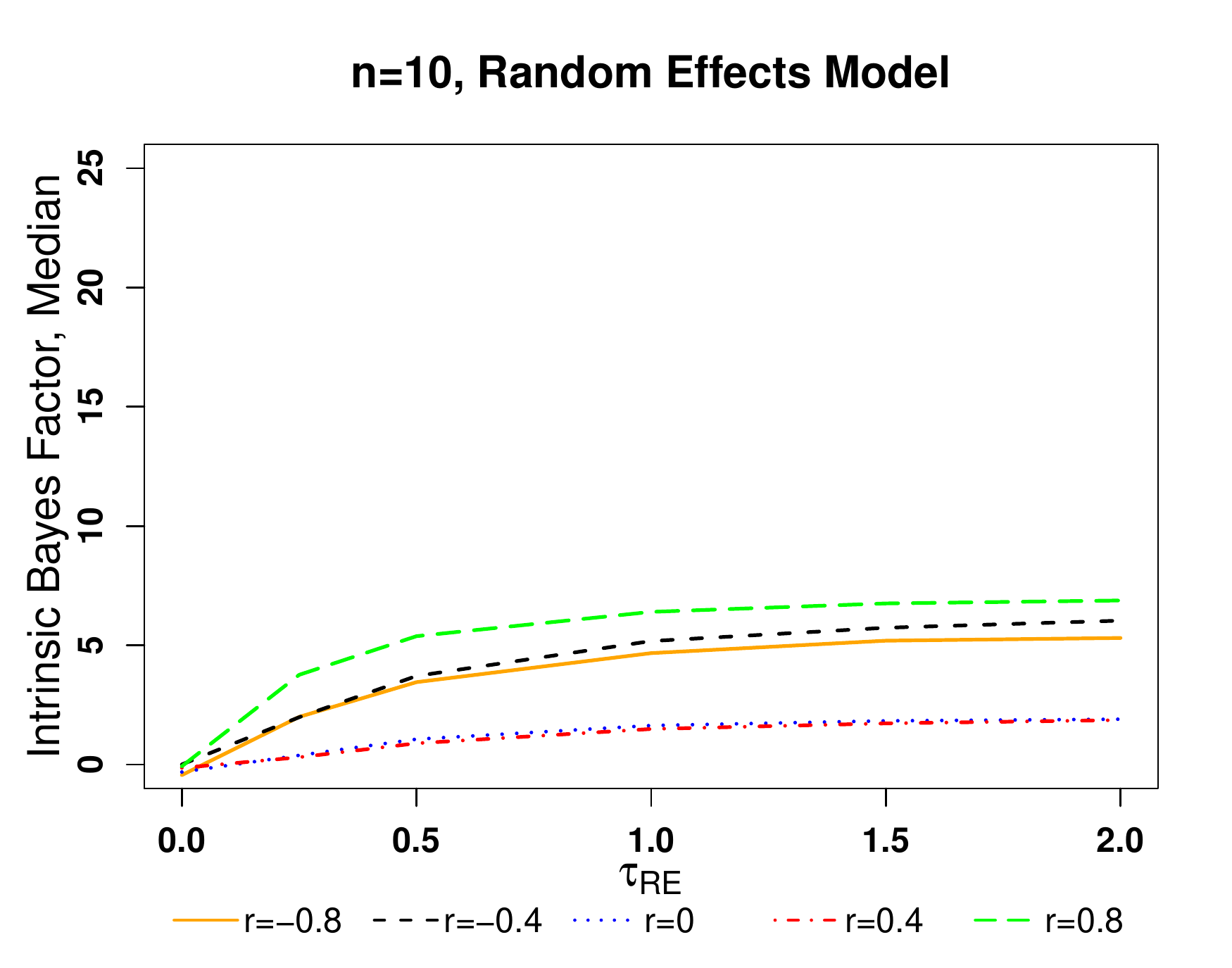}&\includegraphics[width=8cm]{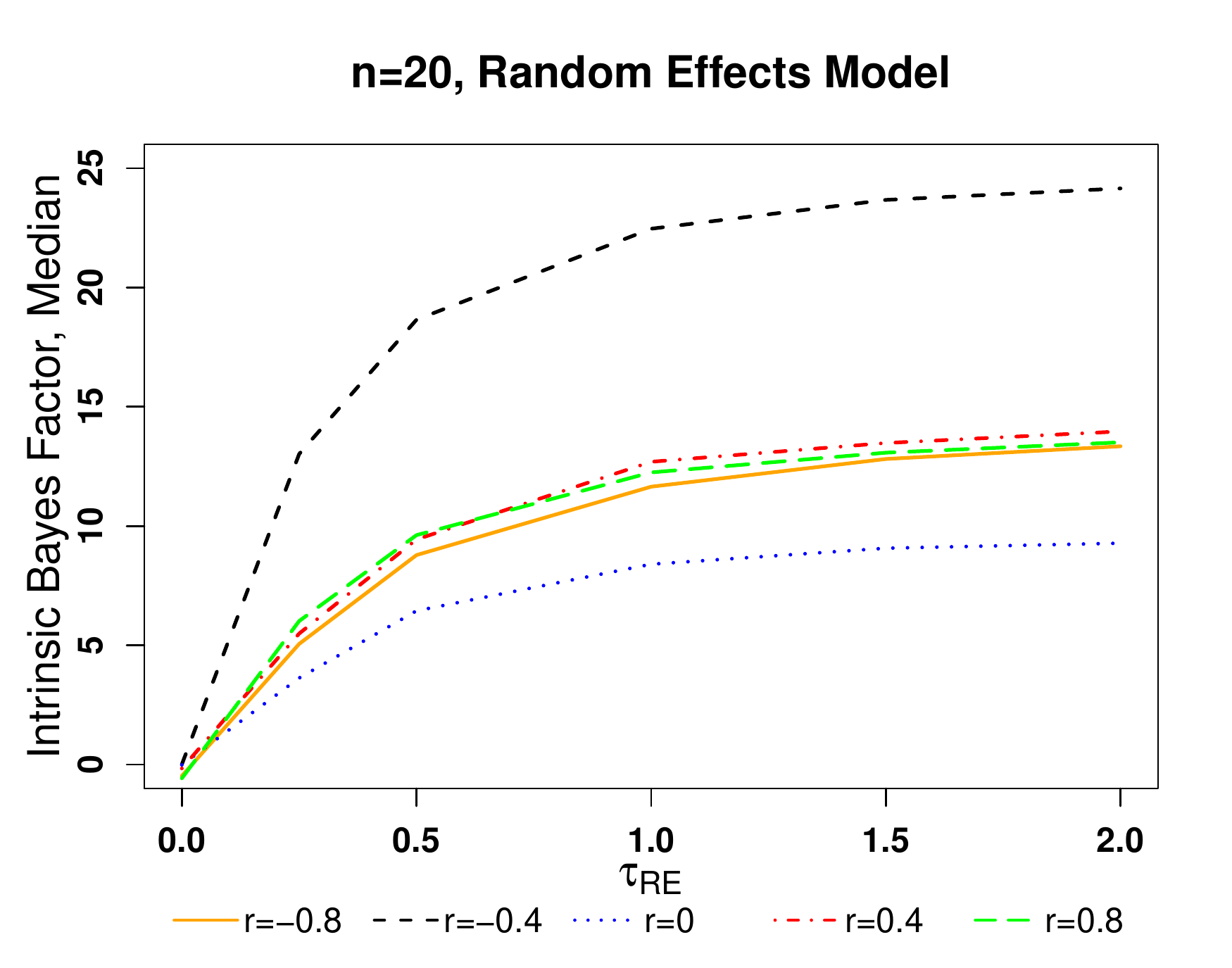}\\
\end{tabular}
 \caption{Median intrinsic Bayes factor \eqref{mIBF_RE_LS} for comparing the random effects model to the location-scale model as a function of $\tau_{RE}$ when the reference prior is employed. We set $n\in\{5,10,15,20\}$ and $r \in \{-0.8,-0.4,0,0.4,0.8\}$. The data were drawn from the random effects model (see, Scenario 2).}
\label{fig:mIBF_RE}
 \end{figure}

\begin{figure}[h!t]
\centering
\begin{tabular}{cc}
\includegraphics[width=8cm]{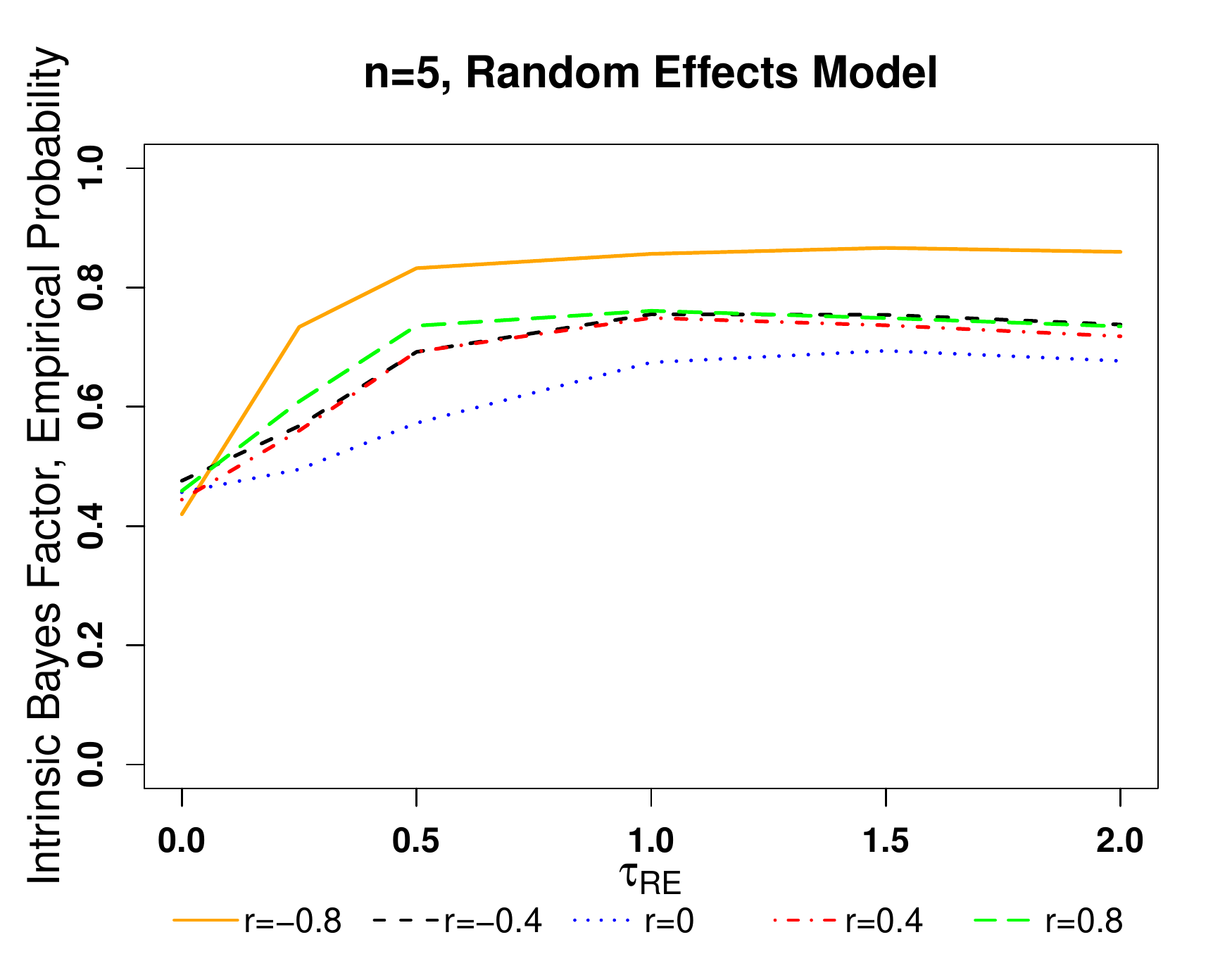}&\includegraphics[width=8cm]{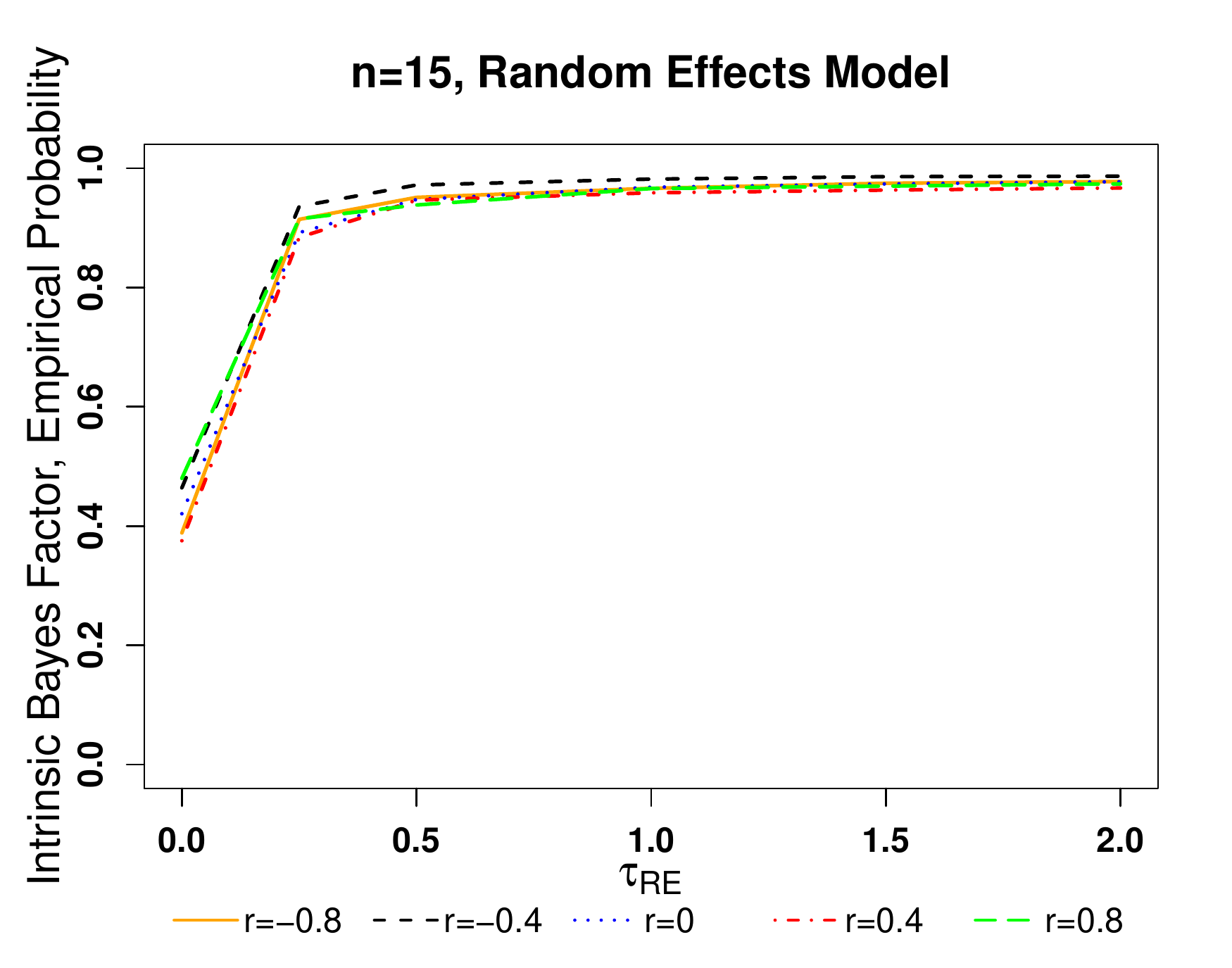}\\
\includegraphics[width=8cm]{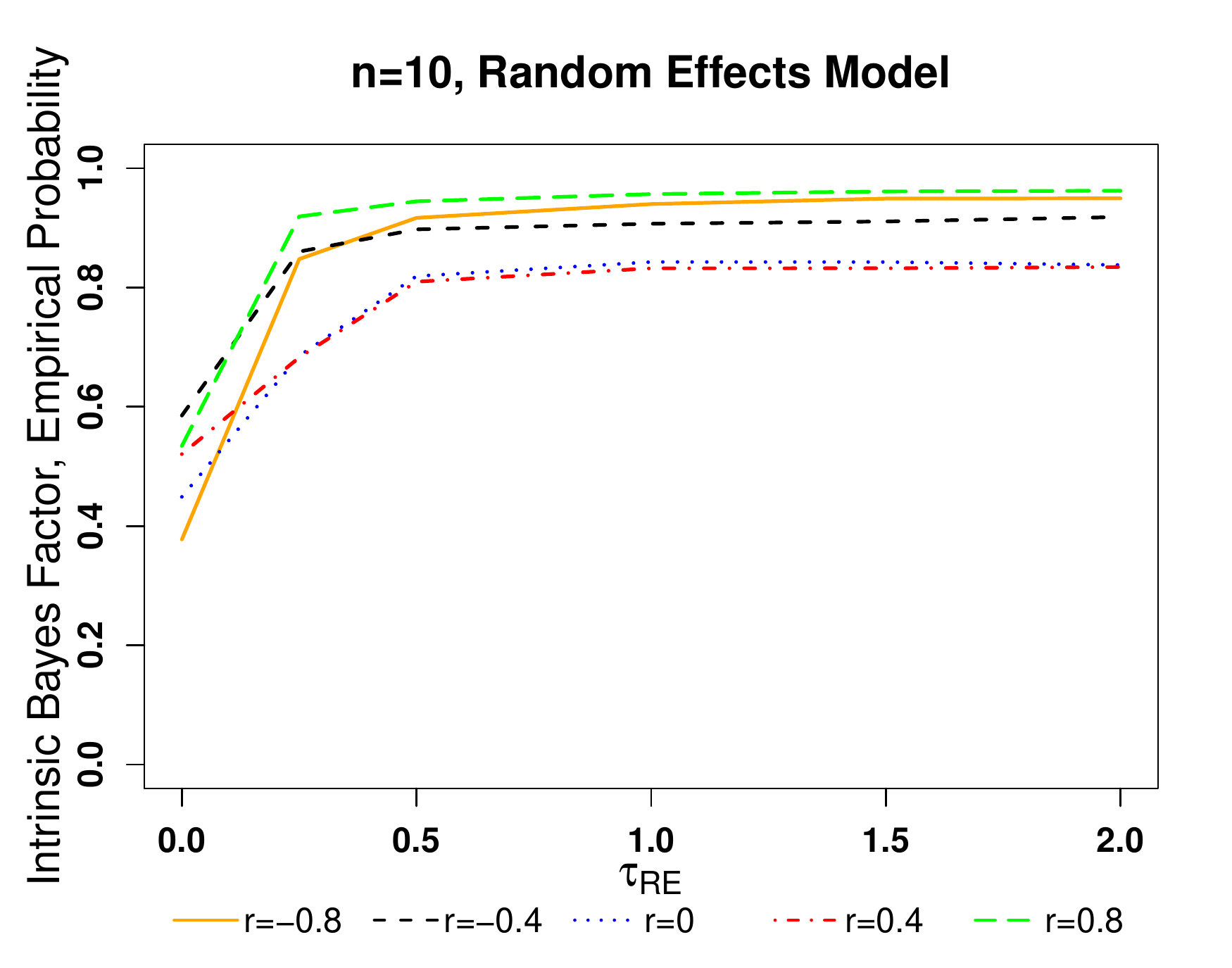}&\includegraphics[width=8cm]{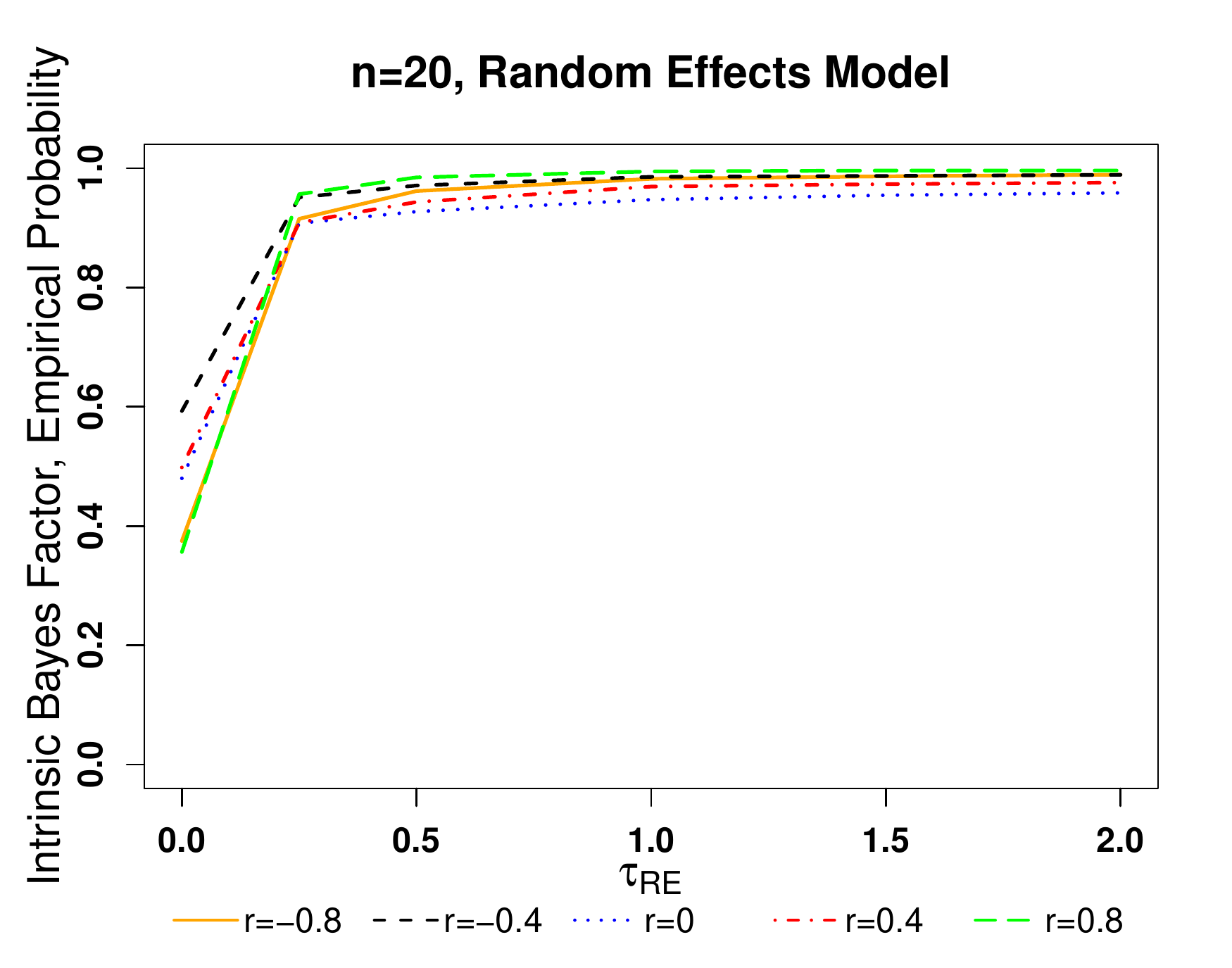}\\
\end{tabular}
 \caption{Empirical probability intrinsic Bayes factor \eqref{epIBF_RE_LS} for comparing the random effects model to the location-scale model as a function of $\tau_{RE}$ when the reference prior is employed. We set $n\in\{5,10,15,20\}$ and $r \in \{-0.8,-0.4,0,0.4,0.8\}$. The data were drawn from the random effects model (see, Scenario 2).}
\label{fig:epIBF_RE}
 \end{figure}

The following two simulation scenarios are considered:
\begin{itemize}
\item \textbf{Scenario 1}: to draw a sample from the location-scale model $\mathbf{x} = \mu\mathbf{1}+\boldsymbol{\varepsilon}_{LS}$ where $\boldsymbol{\varepsilon}_{LS} \sim N(\mathbf{0},\tau_{LS}^2\mathbf{U})$;
\item \textbf{Scenario 2}: to draw a sample from the random effects model $\mathbf{x} = \mu\mathbf{1}+\boldsymbol{\lambda}_{RE}+\boldsymbol{\varepsilon}_{RE}$ where $\boldsymbol{\varepsilon}_{RE} \sim N(\mathbf{0},\mathbf{U})$ and $\boldsymbol{\lambda}_{RE} \sim N(\mathbf{0},\tau_{RE}^2\mathbf{I})$.
\end{itemize}

In both scenarios we set $\mu=0$ and consider several sample sizes $n \in \{5,10,15,20\}$. The square roots of the diagonal elements of the matrix $\mathbf{U}=(u_{ij})_{i,j=1,...,n}$ are drawn from the uniform distribution on $[0,1]$, while its nondiagonal elements of $\mathbf{U}$ are set to $u_{ij}=r^{|i-j|}\sqrt{u_{ii}}\sqrt{u_{jj}}$ for $i,j \in 1,...,n$ and $i \neq j$ with $r \in \{-0.8,-0.4,0,0.4,0.8\}$. Several values of $\tau_{LS}$ and $\tau_{RE}$ are considered, namely $\tau_{LS} \in \{1.0,1.2,1.5,2.0,2.5,3.0\}$ and $\tau_{RE} \in \{0.0,0.25,0.5,1.0,1.5,2.0\}$. The results in Figures \ref{fig:aIBF_LS} to \ref{fig:epIBF_RE} are based on $10000$ independent repetitions. For $\tau_{LS}=1$ and $\tau_{RE}=0$, the two models coincide and they both correspond to the case when the dark uncertainty is absent.

Figures \ref{fig:aIBF_LS} to \ref{fig:epIBF_LS} present the values of the average IBF, the median IBF, and the empirical probability IBF defined in \eqref{aIBF_RE_LS}, \eqref{mIBF_RE_LS}, and \eqref{epIBF_RE_LS}, respectively, and computed for the data generated from the location-scale model (Scenario 1). Similar values obtained under Scenario 2 where the data were drawn from the random effects model are depicted in Figures \ref{fig:aIBF_RE} to \ref{fig:epIBF_RE}. In almost all of the considered cases we observe that the average IBF and the median IBF are negative as well as the empirical probability IBF is smaller than $0.5$ when the data were generated from the location-scale model, and they are positive and the empirical probability IBF is larger than $0.5$ when the data were drawn from the random effects model, thus supporting the model selection based on the IBF factors defined in Section \ref{sec:IBF-LSM-REM}. Some minor deviations from this observation are present only when $n=5$, the data were simulated from the location-scale model, and the values of $\tau_{LS}$ are not large for $r \neq 0$. Interestingly, when $n=5$ and $r=0.0$, then the random effects model is preferred independently whether the data were generated following Scenario 1 or Scenario 2.

In general, we conclude that the model selection with intrinsic Bayes factor detects considerably often the random effects model when it is the true model than the location-scale model when it is the true model. This finding is in line with the results of \citet{bodnar2016evaluation}, who showed that the random effects model is more robust for the model misspecifications. As the sample size increases the performance of the three considered model selection criteria based on the IBF improves. The average IBF and the median IBF become larger in the absolute values, while the empirical probability IBF is close to one for the Scenario 2 and close to 0 for Scenario 1. Already for $n=10$, the values of Figure \ref{fig:epIBF_RE} indicate that the random effects model is correctly chosen in almost all of the considered values when $r \in \{-0.8, -0.4, 0.8\}$ and in $80\%$ for $r \in \{0,0.4\}$, while the probability of correctly specifying the location-scale model is between 70\% and 80\% being slightly smaller for $r=-0.4$. The results for $n=15$ and $n=20$ are even more stronger. For example, the random effects model is chosen with probability larger $0.9$ in both cases $r=0$ and $r=0.4$ for all considered values $\tau_{RE}>0$ (cf., Figure \ref{fig:epIBF_RE}). Finally, we note that only a minor impact of the correlation coefficient $r$ is present on the computed values of the intrinsic Bayes factors.

\section{Modeling dark uncertainty in the measurements of the Newtonian constant of gravitation }\label{sec:emp}

In the year 1687 Newton published his famous $principia$ in which he presented his thesis on the law of general gravitation. Newton proposed that the force that cause an apple to fall from a tree is the same kind of force that keep the moon to orbit the earth, this would be the force of gravitational attraction. In most textbooks in physics this law is written by an equation $F=G\frac{m_1m_2}{r^2}$. The constant of proportionality, $G$, is called the general constant of gravitation of the Newtonian constant of gravitation, and it is a natural constant, i.e., this constant is the same throughout the whole universe. It determines the strength of the gravitational force, given mass and distance. If we imagine two objects being put out in space far away from any stars and planets at a distance of $1\,\text{m}$ from each other and they both have the mass $1\,\text{kg}$, then the objects exert a gravitational force on each other by exactly $F=G\approx 0.000\,000\,000\,066\,7\, \text{N} = 6.67\cdot 10^{-11} \, \text{N}$.


\begin{table}[h!t]
\begin{center}
\begin{tabular}{ l | c | c }
	\hline\hline
	Study & Measurement (G$/10^{-11}$) & Uncertainty (G$/10^{-11}$) \\ \hline
	NIST-82	& 6.67248 & 0.00043 \\
	TR\&D-96 & 6.6729 & 0.00050 \\
	LANL-97	& 6.67398 & 0.00070 \\
	UWash-00 & 6.674255 & 0.000092 \\
	BIPM-01	& 6.67559 & 0.00027 \\
	UWup-02	& 6.67422 & 0.00098 \\
	MSL-03	& 6.67387 & 0.00027 \\
	HUST-05	& 6.67222 & 0.00087 \\
	UZur-06	& 6.67425 & 0.00012 \\
	HUST-09	& 6.67349 & 0.00018 \\
	JILA-10	& 6.67260 & 0.00025 \\
	BIPM-14	& 6.67554 & 0.00016 \\
	LENS-14	& 6.67191 & 0.00099 \\
	UCI-14	& 6.67435 & 0.00013 \\
	HUST-TOS-18	& 6.674184 & 0.000078 \\
	HUST-AAF-18	& 6.674484 & 0.000078 \\
	\hline \hline
\end{tabular}
\end{center}
\caption{Measurement results for the Newtonian constant of gravitation G together with uncertainties (data are from Figure 1 in \citet{bodnar2020bayesian}).}\label{tab:G}
\end{table}

Sixteen measurements of the Newtonian constant of gravitation together with their uncertainties are provided in Table \ref{tab:G} and are shown in Figure \ref{fig:G}, from which it might be concluded about the presence of dark uncertainty. For capturing the effect of heterogeneity, one can use the location-scale model or the random effects model, which have already been applied for these purposes (\citet{codata2016}, \citet{bodnar2020bayesian}).

\begin{figure}[h!t]
	\begin{center}
		\includegraphics[scale=0.7, clip=true, trim=0 10 0 10]{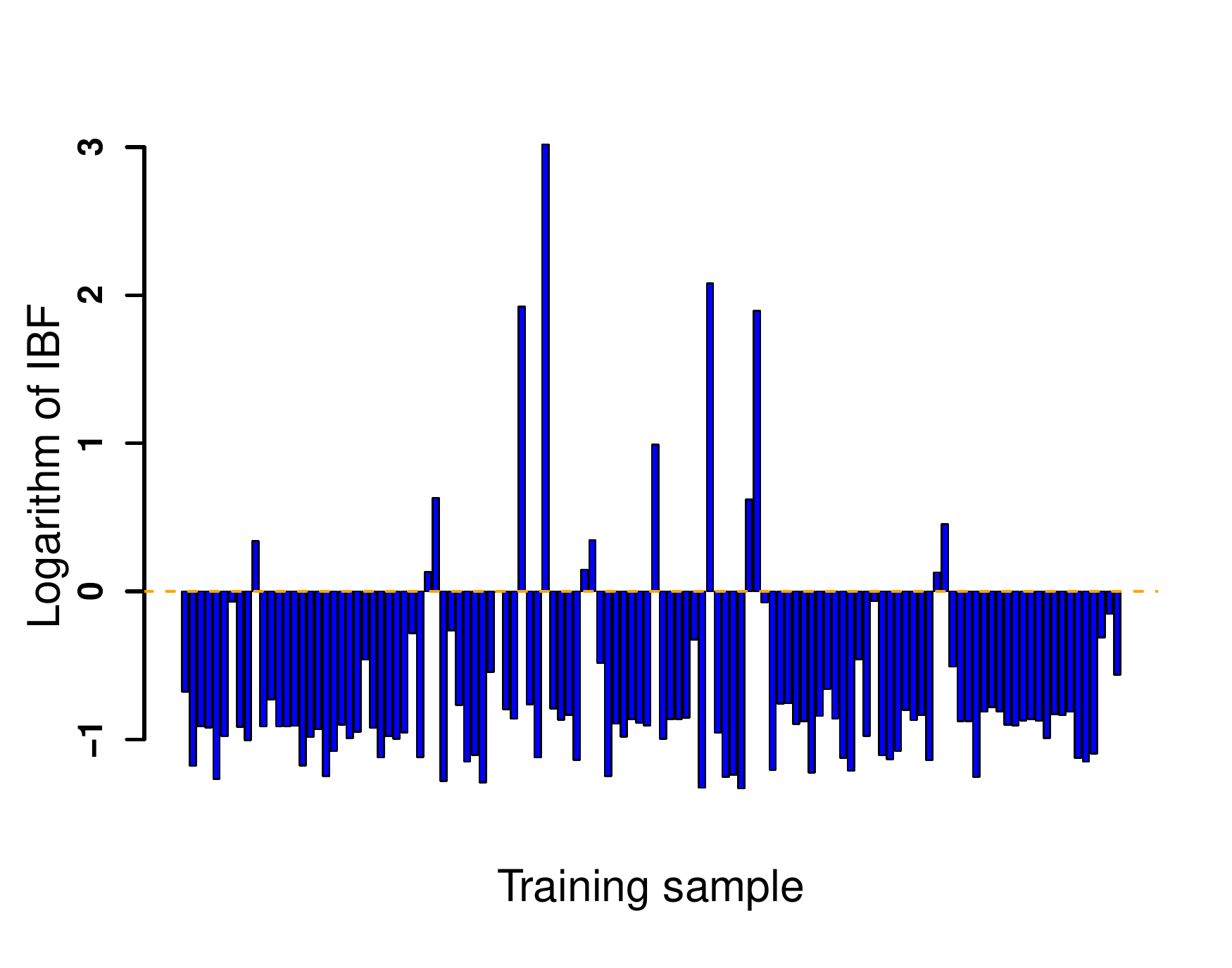}
	\end{center}
	\caption{Logarithm of intrinsic Bayes factors computed for each of the $L=120$ minimal training samples of size $m=2$ by using data for the Newtonian constant of gravitation from Table \ref{tab:G}.} \label{fig:G-IBF}
\end{figure}

In order to make a preference for one of the two models, we assign a reference prior to the parameters of two models and perform the Bayesian model selection based on the intrinsic Bayesian factor. The application of a non-informative reference prior is motivated by the absence of information about the parameters of the two models which can be used to determine an informative prior. Using that the size of the minimal training sample is $m=2$, we computed the intrinsic Bayes factor for all 120 possible specifications of the training sample consisting of two measurement results.

The resulting values of the IBF for comparing the random effects model to the location-scale model are depicted in Figure \ref{fig:G-IBF}. We observe that in most of the cases the IBF is negative meaning that the location-scale model is preferable. Only for 13 training samples out of 120 possible training samples the random effects model is chosen. This leads to the value of the empirical probability IBF equal to $epB_{M_{RE}M_{LS}}^I=0.1083$. The other two measures for Bayesian model selection discussed in Section \ref{sec:IBF-LSM-REM} are negative and they are given by
\begin{equation}
	aB_{M_{RE}M_{LS}}^I=-0.6814 \quad \text{and} \quad mB_{M_{RE}M_{LS}}^I=-0.8777,
\end{equation}
which support the application of the location-scale model to fit the heterogeneity in the measurements of the Newtonian constant of gravitation.

\section{Summary} \label{sec:sum}

In many applications the variability of the individual studies, which are pooled together to determine the overall mean value, cannot be explained by the reported variabilities of each studies. This leads to the conclusion of the presence of heterogeneity, also known as the dark uncertainty. Two mostly used models for the dark uncertainty are the location-scale model and the random effects models. The location-scale model is usually applied in connection to the Birge ratio method for adjusting the measurements of the fundamental constants in physics and chemistry, while the random effects model is the classical tool used for meta-analysis in medicine.

These two models are non-nested and describe the heterogeneity from different perspectives. While the random effects model suggests to adjust the reported variabilities of individual studies by adding a constant to the reported uncertainties, the location-scale model adjusts the underrated individual uncertainties by a multiplicative constant. Since the two approaches are non-nested statistical models, the methods of the frequentist statistics, like the likelihood ratio test cannot be used for choosing the model which provides a better fit to data. In such a situation methods of Bayesian statistics for model selection should be opted. Moreover, in most of applications no additional information is provided for the model parameters and thus noninformative priors should be used in the derivation of Bayesian inference procedures. This has a strong impact on the Bayesian model selection, since the noninformative priors are usually improper, thus leading to the improper marginal distributions of data needed in the computation of the Bayes factor, a rule for choosing one of two competitive models in Bayesian statistics.

The solution to the latter problem was suggested in the seminal paper of \citet{berger1996intrinsic}, who introduced the intrinsic Bayes factor for improper noninformative priors. Endowing the parameters of the location-scale model and of the random effects model with reference prior (\citet{BergerBernardo1992c}), we derive the expression of the intrinsic Bayes factor for comparing the random effects model to the location-scale model. In the case of the location-scale model the analytical expression of the marginal distributions of the whole data and of the training data are derived, while one-dimensional integral presentations of the marginal distribution are obtained in the case of the random effects model. These integrals can be computed numerically using Simpson's rule.

The performance of the suggested Bayesian model selection procedure has been investigated within an extensive simulation study. It was found that the procedure can distinguish between two models when the data were generated from the random effects model is generated already for small values of heterogeneity parameter and when the sample size is small, like five observations are only present. In contrast, when the data were generated from the location-scale model, then the derived model selection procedure requires larger sample size and larger values of the heterogeneity parameter to detect the true model. Such findings are in line with the previous results documented in \citet{bodnar2016evaluation}, where it is shown that the random effects model is more robust to the model misspecification than the location-scale model. In particular, it was shown numerically that even the data were drawn from the location-scale model, the random effects model is still persistent making a good estimate of the overall mean. That was not the case when the data were generated from the random effects model and the location-scale model is applied.

The derived theoretical findings are applied to the measurement results of the Newtonian constant of gravitation. The results of the empirical study support the application of the location-scale model for the computation of the Newtonian constant of gravitation. Both the average IBF and the median IBF are significantly smaller than zero, while the empirical probability IBF was less than 11\% indicating that in majority of the cases related to the specification of the training sample the location-scale model ia favored to the random effects model. As such our findings are supportive to the application of the Birge ratio method for the adjustment of the fundamental physical constants as it is currently used in the computation of the CODATA 2018 values of fundamental constants (see, \citet{codata2018}). Finally, it can to be noted that the obtained results only document that the location-scale model is favored compared to the random effects model, but they do not answer the general question what is the best approach for data inconsistency or heterogeneous data.

\section*{Acknowledgement}

This research was partially supported by National Institute of Standards and Technology (NIST) Exchange Visitor Program. The first author is grateful to the Statistical Engineering Division of National Institute of Standards and Technology (NIST) for providing an excellent and inspiring environment for research. Olha Bodnar also acknowledges valuable support from the internal grand (R\"{o}rlig resurs) of the \"{O}rebro University.

\section{Appendix}\label{sec:app}
In this section, the proofs of Theorems \ref{th1} and \ref{th2} are presented.

\begin{proof}[Proof of Theorem \ref{th1}:]
\begin{enumerate}[(i)]
\item In the case of the whole date, the marginal distribution is obtained by integration of the parameters $\mu$ and $\tau_{LS}$ in the joint probability density function $f(\mathbf{x},\mu,\tau_{LS}|M_{LS})$. It holds that
\begin{eqnarray*}
m(\mathbf{x}|M_{LS})&=&\int_{0}^{\infty}\int_{-\infty}^{\infty}f(\mathbf{x}|\mu,\tau_{LS},M_{LS})\pi^N(\mu,\tau_{LS}) \ d\mu \ d\tau_{LS} \\ 	
&=& \int_{0}^{\infty}\int_{-\infty}^{\infty}\frac{\big(\mathsf{det}(\mathbf{U})\big)^{-\frac{1}{2}}}{(2\pi)^{n/2}}\tau_{LS}^{-n}\mathsf{exp}
\bigg(-\frac{(\mathbf{x}-\mu\mathbf{1})^T\mathbf{U}^{-1}(\mathbf{x}-\mu\mathbf{1})}{2\tau_{LS}^2}\bigg)\frac{1}{\tau_{LS}} \ d\mu \ d\tau_{LS} \\
&=&\frac{\big(\mathsf{det}(\mathbf{U})\big)^{-\frac{1}{2}}}{2(2\pi)^{n/2}}\int_{-\infty}^{\infty}
\Bigg(\int_{0}^{\infty}\zeta^{-(\frac{n}{2}+1)}\textmd{exp}
\bigg(-\frac{\frac{(\mathbf{x}-\mu\mathbf{1})^T\mathbf{U}^{-1}(\mathbf{x}-\mu\mathbf{1})}{2}}{\zeta}\bigg)  d\zeta\Bigg) d\mu,
\end{eqnarray*}
where the last equality is obtained by making the transformation $\tau_{LS}^2 = \zeta$ with the Jacobian $\frac{1}{2\sqrt{\zeta}}$.

The inner integral with respect to $\zeta$ has an integrand that is the kernel of the density function of the inverse gamma distribution $\textmd{inv-gamma}(n/2, (\mathbf{x}-\mu\mathbf{1})^T\mathbf{U}^{-1}(\mathbf{x}-\mu\mathbf{1})/2)$. Hence,
\begin{eqnarray*}
	m(\mathbf{x}|M_{LS}) & =&  
\frac{\big(\mathsf{det}(\mathbf{U})\big)^{-\frac{1}{2}}\Gamma\big(\frac{n}{2}\big)}{2\pi^{n/2}} \int_{-\infty}^{\infty}\big((\mathbf{x}-\mu\mathbf{1})^T\mathbf{U}^{-1}(\mathbf{x}-\mu\mathbf{1})\big)^{-n/2} \ d\mu \ .
\end{eqnarray*}
	
Using that
\begin{equation}
	(\mathbf{x}-\mu\mathbf{1})^T\mathbf{U}^{-1}(\mathbf{x}-\mu\mathbf{1})  = \mathbf{x}^T\mathbf{Qx}+\mathbf{1}^T\mathbf{U}^{-1}\mathbf{1}\bigg(\mu-\frac{\mathbf{1}^T\mathbf{U}^{-1}\mathbf{x}}{\mathbf{1}^T\mathbf{U}^{-1}\mathbf{1}}\bigg)^2 \ , \nonumber	
\end{equation}
with $\mathbf{Q}$ defined in \eqref{Q:LSM}, we get
\begin{eqnarray*}
&&\hspace{-1cm}	m(\mathbf{x}|M_{LS})  =  \frac{\big(\mathsf{det}(\mathbf{U})\big)^{-\frac{1}{2}}\Gamma\big(\frac{n}{2}\big)}{2\pi^{n/2}} (\mathbf{x}^T\mathbf{Qx})^{-n/2} \int_{-\infty}^{\infty}\left(1+\frac{\mathbf{1}^T\mathbf{U}^{-1}\mathbf{1}}{\mathbf{x}^T\mathbf{Qx}}
\bigg(\mu-\frac{\mathbf{1}^T\mathbf{U}^{-1}\mathbf{x}}{\mathbf{1}^T\mathbf{U}^{-1}\mathbf{1}}\bigg)^2\right)^{-n/2} d\mu \\
 & =& \frac{\Gamma(\frac{n-1}{2})\big(\mathbf{x}^T\mathbf{Q}\mathbf{x}\big)^{-\frac{n-1}{2}}}{\big(\mathsf{det}(\mathbf{U})\big)^{1/2} 2\pi^{\frac{n-1}{2}}\sqrt{\mathbf{1}^T\mathbf{U}^{-1}\mathbf{1}}}\\
 &\times&\int_{-\infty}^{\infty}\frac{\Gamma\big(\frac{n}{2}\big)}{\sqrt{\pi}\Gamma\big(\frac{n-1}{2}\big)}\sqrt{\frac{\mathbf{1}^T\mathbf{U}^{-1}\mathbf{1}}{\mathbf{x}^T\mathbf{Qx}}}
 \left(1+\frac{1}{n-1}\frac{(n-1)\mathbf{1}^T\mathbf{U}^{-1}\mathbf{1}}{\mathbf{x}^T\mathbf{Qx}}
\bigg(\mu-\frac{\mathbf{1}^T\mathbf{U}^{-1}\mathbf{x}}{\mathbf{1}^T\mathbf{U}^{-1}\mathbf{1}}\bigg)^2\right)^{-n/2} d\mu\\
&=&\frac{\Gamma(\frac{n-1}{2})\big(\mathbf{x}^T\mathbf{Q}\mathbf{x}\big)^{-\frac{n-1}{2}}}{\big(\mathsf{det}(\mathbf{U})\big)^{1/2} 2\pi^{\frac{n-1}{2}}\sqrt{\mathbf{1}^T\mathbf{U}^{-1}\mathbf{1}}},
\end{eqnarray*}
where the last equality follows by recognizing the density function of the $t$-distribution under the integral in the second line.

\item The statement of the second part of the theorem is obtained by following the proof of part (i) and noting that when $m=1$ then $m(x_i|M_{LS})$ becomes a constant and, thus, $m(x_i|M_{LS})$ is not longer a proper density function. When $m=2$ we get
\begin{equation}
	m(\mathbf{x}_{\ell}|M_{LS})=\frac{1}{2\sqrt{(\mathsf{det}(\mathbf{U}_{\ell}))(\mathbf{x}^T\mathbf{Q}_{\ell}\mathbf{x})(\mathbf{1}_2^T\mathbf{U}_{\ell}^{-1}\mathbf{1}_2)}} \ ,
\end{equation}
where $\mathbf{Q}_{\ell}$ is given in \eqref{Ql:LSM}. Since the last expression is a proper density function and the two elements of $\mathbf{x}$ are arbitrary chosen, the theorem is proved.
\end{enumerate}
\end{proof}

\begin{proof}[Proof of Theorem \ref{th2}:]
\begin{enumerate}[(i)]	
\item In the case of the whole data the marginal distribution is obtained by integration of the parameters $\mu$ and $\tau_{RE}$ in the joint probability density function $f(\mathbf{x},\mu,\tau_{RE}|M_{RE})$ and it is expressed as
\begin{eqnarray*}
&&m(\mathbf{x}|M_{RE})=\int_{0}^{\infty}\int_{-\infty}^{\infty}f(\mathbf{x}|\mu,\tau_{RE},M_{RE})\pi^N(\tau_{RE}) \ d\mu \ d\tau_{RE} \\ 	
&=& \int_{0}^{\infty}\int_{-\infty}^{\infty}\frac{\big(\mathsf{det}(\mathbf{U}+\tau_{RE}^2\mathbf{I})\big)^{-\frac{1}{2}}}{(2\pi)^{n/2}}
\mathsf{exp}\bigg(-\frac{1}{2}(\mathbf{x}-\mu\mathbf{1})^T\big(\mathbf{U}+\tau_{RE}^2\mathbf{I}\big)^{-1}(\mathbf{x}-\mu\mathbf{1})\bigg) \pi^N(\tau_{RE}) \ d\mu \ d\tau_{RE} , \end{eqnarray*}
where $\pi^N(\tau_{RE})$ is given in \eqref{prior:REM}.

Using the identity
\begin{eqnarray*}
	&&\mathsf{exp}\bigg(-\frac{1}{2}(\mathbf{x}-\mu\mathbf{1})^T\big(\mathbf{U}+\tau_{RE}^2\mathbf{I}\big)^{-1}(\mathbf{x}-\mu\mathbf{1})\bigg) \\
	&=& \mathsf{exp}\Big(-\frac{1}{2}\mathbf{x}^T\mathbf{Q}(\tau_{RE}^2)\mathbf{x}\Big)\mathsf{exp}\Bigg(-\frac{1}{2}\mathbf{1}^T\big(\mathbf{U}+\tau_{RE}^2\mathbf{I}\big)^{-1}
\mathbf{1}\bigg(\mu-\frac{\mathbf{1}^T\big(\mathbf{U}+\tau_{RE}^2\mathbf{I}\big)^{-1}\mathbf{x}}{\mathbf{1}^T\big(\mathbf{U}+\tau_{RE}^2\mathbf{I}\big)^{-1}\mathbf{1}}\bigg)^2\Bigg),
\end{eqnarray*}
with $\mathbf{Q}(\tau_{RE}^2)$ defined in \eqref{Q:REM} and integrating over $\mu$, we get
\begin{eqnarray*}
m(\mathbf{x}|M_{RE})&=&\int_{0}^{\infty} \frac{\big(\mathsf{det}(\mathbf{U}+\tau_{RE}^2\mathbf{I})\big)^{-\frac{1}{2}}}{(2\pi)^{(n-1)/2}} \frac{\mathsf{exp}\Big(-\frac{1}{2}\mathbf{x}^T\mathbf{Q}(\tau_{RE}^2)\mathbf{x}\Big) }{\sqrt{\mathbf{1}^T\big(\mathbf{U}+\tau_{RE}^2\mathbf{I}\big)^{-1}\mathbf{1}}} \pi^N(\tau_{RE})
\ d\tau_{RE} ,
\end{eqnarray*}
which is the statement of the first part of the theorem.

\item Part (ii) of the theorem follows from the proof of the first part.
\end{enumerate}
\end{proof}

{\footnotesize
\bibliography{Bibfile2021-03-01}

\begin{thebibliography}{}

\bibitem[Ades et~al., 2005]{AdesHiggins2005}
Ades, A.~E., Lu, G., and Higgins, J. (2005).
\newblock The interpretation of random-effects meta-analysis in decision
  models.
\newblock {\em Medical Decision Making}, 25(6):646--654.

\bibitem[Alighanbari et~al., 2020]{alighanbari2020precise}
Alighanbari, S., Giri, G., Constantin, F.~L., Korobov, V., and Schiller, S.
  (2020).
\newblock Precise test of quantum electrodynamics and determination of
  fundamental constants with {HD+} ions.
\newblock {\em Nature}, 581(7807):152--158.

\bibitem[Berger and Bernardo, 1992]{BergerBernardo1992c}
Berger, J. and Bernardo, J.~M. (1992).
\newblock On the development of reference priors.
\newblock In Bernardo, J.~M., Berger, J., Dawid, A.~P., and Smith, A. F.~M.,
  editors, {\em Bayesian Statistics}, volume~4, pages 35--60. Oxford:
  University Press.

\bibitem[Berger et~al., 2009]{berger2009formal}
Berger, J.~O., Bernardo, J.~M., and Sun, D. (2009).
\newblock The formal definition of reference priors.
\newblock {\em The Annals of Statistics}, 37(2):905--938.

\bibitem[Berger and Pericchi, 1996]{berger1996intrinsic}
Berger, J.~O. and Pericchi, L.~R. (1996).
\newblock The intrinsic {B}ayes factor for model selection and prediction.
\newblock {\em Journal of the American Statistical Association},
  91(433):109--122.

\bibitem[Berger and Pericchi, 2001]{berger2001objective}
Berger, J.~O. and Pericchi, L.~R. (2001).
\newblock Objective {B}ayesian methods for model selection: {I}ntroduction and
  comparison.
\newblock In {\em Model Selection}, pages 135--207. Hayward, CA: Institute of
  Mathematical Statistics.

\bibitem[Birge, 1932]{Birge1932}
Birge, R.~T. (1932).
\newblock The calculation of errors by the method of the least squares.
\newblock {\em Physical Reviev}, 40(2):207--227.

\bibitem[Bodnar and Elster, 2014a]{bodnar2014analytical}
Bodnar, O. and Elster, C. (2014a).
\newblock Analytical derivation of the reference prior by sequential
  maximization of {S}hannon's mutual information in the multi-group parameter
  case.
\newblock {\em Journal of Statistical Planning and Inference}, 147:106--116.

\bibitem[Bodnar and Elster, 2014b]{bodnar2014adjustment}
Bodnar, O. and Elster, C. (2014b).
\newblock On the adjustment of inconsistent data using the {B}irge ratio.
\newblock {\em Metrologia}, 51(5):516.

\bibitem[Bodnar and Elster, 2020]{bodnar2020assessing}
Bodnar, O. and Elster, C. (2020).
\newblock Assessing laboratory effects in key comparisons with two transfer
  standards measured in two petals: A {B}ayesian approach.
\newblock In {\em 13th International Workshop on Intelligent Statistical
  Quality Control 2019, IWISQC 2019, 12 August 2019 through 14 August 2019},
  pages 1--18. City University of Hong Kong.

\bibitem[Bodnar et~al., 2016a]{bodnar2016evaluation}
Bodnar, O., Elster, C., Fischer, J., Possolo, A., and Toman, B. (2016a).
\newblock Evaluation of uncertainty in the adjustment of fundamental constants.
\newblock {\em Metrologia}, 53(1):S46.

\bibitem[Bodnar et~al., 2017]{bodnar2017bayesian}
Bodnar, O., Link, A., Arendack{\'a}, B., Possolo, A., and Elster, C. (2017).
\newblock Bayesian estimation in random effects meta-analysis using a
  non-informative prior.
\newblock {\em Statistics in Medicine}, 36(2):378--399.

\bibitem[Bodnar et~al., 2016b]{BodnarLinkElster2015}
Bodnar, O., Link, A., and Elster, C. (2016b).
\newblock Objective {B}ayesian inference for a generalized marginal random
  effects model.
\newblock {\em Bayesian Analysis}, 11(1):25--45.

\bibitem[Bodnar et~al., 2020]{bodnar2020bayesian}
Bodnar, O., Muhumuza, R.~N., and Possolo, A. (2020).
\newblock Bayesian inference for heterogeneity in meta-analysis.
\newblock {\em Metrologia}, 57(6):064004.

\bibitem[Brockwell and Gordon, 2001]{brockwell2001}
Brockwell, S.~E. and Gordon, I.~R. (2001).
\newblock A comparison of statistical methods for meta-analysis.
\newblock {\em Statistics in Medicine}, 20(6):825--840.

\bibitem[Browne and Draper, 2006]{BrowneDraper06}
Browne, W.~J. and Draper, D. (2006).
\newblock A comparison of {B}ayesian and likelihood-based methods for fitting
  multilevel models.
\newblock {\em Bayesian Analysis}, 1(3):473--514.

\bibitem[Clarke and Yuan, 2004]{ClarkeYuan2004}
Clarke, B. and Yuan, A. (2004).
\newblock Partial information reference priors: {D}erivation and
  interpretations.
\newblock {\em Journal of Statistical Planning and Inference}, 123(2):313--345.

\bibitem[Cochran, 1937]{Cochran37}
Cochran, W.~G. (1937).
\newblock Problems arising in the analysis of a series of similar experiments.
\newblock {\em Journal of the Royal Statistical Society - Supplement},
  4:102--118.

\bibitem[Cochran, 1954]{Cochran54}
Cochran, W.~G. (1954).
\newblock The combination of estimates from different experiments.
\newblock {\em Biometrics}, 10:109--129.

\bibitem[Fern\'andez and Steel, 1999]{FernandezSteel1999}
Fern\'andez, C. and Steel, M. F.~J. (1999).
\newblock Reference priors for the general location-scale model.
\newblock {\em Statistics \& Probability Letters}, 43:377--384.

\bibitem[Gelman, 2006]{Gelman06}
Gelman, A. (2006).
\newblock Prior distributions for variance parameters in hierarchical models.
\newblock {\em Bayesian Analysis}, 1:515--533.

\bibitem[Givens and Hoeting, 2012]{givens2012computational}
Givens, G.~H. and Hoeting, J.~A. (2012).
\newblock {\em Computational Statistics}, volume 710.
\newblock John Wiley \& Sons.

\bibitem[Guolo and Varin, 2017]{guolo2017random}
Guolo, A. and Varin, C. (2017).
\newblock Random-effects meta-analysis: {T}he number of studies matters.
\newblock {\em Statistical Methods in Medical Research}, 26(3):1500--1518.

\bibitem[Hardy and Thompson, 1998]{hardy1998detecting}
Hardy, R.~J. and Thompson, S.~G. (1998).
\newblock Detecting and describing heterogeneity in meta-analysis.
\newblock {\em Statistics in Medicine}, 17(8):841--856.

\bibitem[Higgins et~al., 2009]{Higgins2009}
Higgins, J., Thompson, S.~G., and Spiegelhalter, D.~J. (2009).
\newblock A re-evaluation of random-effects meta-analysis.
\newblock {\em Journal of the Royal Statistical Society: Ser. A (Statistics in
  Society)}, 172(1):137--159.

\bibitem[Hill, 1965]{Hill65}
Hill, B.~M. (1965).
\newblock Inference about variance components in the one-way model.
\newblock {\em Journal of the American Statistical Association}, 60:806--825.

\bibitem[Jeffreys, 1946]{Jeffreys1946}
Jeffreys, H. (1946).
\newblock An invariant form for the prior probability in estimation problems.
\newblock {\em Proceedings of the Royal Society A}, 186:453--461.

\bibitem[Jones et~al., 2018]{jones2018use}
Jones, H.~E., Ades, A., Sutton, A.~J., and Welton, N.~J. (2018).
\newblock Use of a random effects meta-analysis in the design and analysis of a
  new clinical trial.
\newblock {\em Statistics in Medicine}, 37(30):4665--4679.

\bibitem[Kacker, 2004]{Kacker2004}
Kacker, R.~N. (2004).
\newblock Combining information from interlaboratory evaluations using a random
  effects model.
\newblock {\em Metrologia}, 41:132--136.

\bibitem[Lambert et~al., 2005]{Lambert2005}
Lambert, P.~C., Sutton, A.~J., Burton, P.~R., Abrams, K.~R., and Jones, D.~R.
  (2005).
\newblock How vague is vague? {A} simulation study of the impact of the use of
  vague prior distributions in mcmc using winbugs.
\newblock {\em Statistics in Medicine}, 24(15):2401--2428.

\bibitem[Laplace, 1812]{Laplace1812}
Laplace, P.~S. (1812).
\newblock {\em Th\'{e}orie Analitique des Probabilit\'{e}s}.
\newblock Paris: Courcier.

\bibitem[Mandel and Paule, 1970]{mandel-1970}
Mandel, J. and Paule, R. (1970).
\newblock Interlaboratory evaluation of a material with unequal numbers of
  replicates.
\newblock {\em Analytical Chemistry}, 42(11):1194--1197.

\bibitem[Mohr et~al., 2016]{codata2016}
Mohr, P.~J., Newell, D.~B., and Taylor, B.~N. (2016).
\newblock {CODATA} recommended values of the fundamental physical constants:
  2014.
\newblock {\em Reviews of Modern Physics}, 88(3):035009.

\bibitem[Newell et~al., 2018]{newell2018codata}
Newell, D.~B., Cabiati, F., Fischer, J., Fujii, K., Karshenboim, S.~G.,
  Margolis, H.~S., de~Mirandes, E., Mohr, P.~J., Nez, F., and Pachucki, K.
  (2018).
\newblock The {CODATA} 2017 values of h, e, k, and {$N_A$} for the revision of
  the {SI}.
\newblock {\em Metrologia}, 55(1):L13--L16.

\bibitem[O'Hagan, 1995]{o1995fractional}
O'Hagan, A. (1995).
\newblock Fractional bayes factors for model comparison.
\newblock {\em Journal of the Royal Statistical Society: Series B
  (Methodological)}, 57(1):99--118.

\bibitem[Rao, 1997]{Rao1997}
Rao, P. S. R.~S. (1997).
\newblock {\em Variance Components Estimation: Mixed Models, Methodologies, and
  Applications}.
\newblock Chapman and Hall, London.

\bibitem[Ruhkin, 2003]{rukhin-2003}
Ruhkin, A. (2003).
\newblock Two procedures of meta-analysis in clinical trials and
  interlaboratory studies.
\newblock {\em Tatra Mountains Mathematical Publications}, 26(155):155--168.

\bibitem[Rukhin, 2013]{Rukhin2013}
Rukhin, A.~L. (2013).
\newblock Estimating heterogeneity variance in meta-analysis.
\newblock {\em Journal of the Royal Statistical Society: Ser. B}, 75:451--469.

\bibitem[Searle et~al., 2006]{Searle2006}
Searle, S.~R., Casella, G., and Mc~Culloch, C.~E. (2006).
\newblock {\em Variance Components}.
\newblock John Wiley \& Sons, New Jersey.

\bibitem[Thompson and Ellison, 2011]{thompson2011dark}
Thompson, M. and Ellison, S.~L. (2011).
\newblock Dark uncertainty.
\newblock {\em Accreditation and Quality Assurance}, 16(10):483--487.

\bibitem[Tiao and Tan, 1965]{TiaoTan1965}
Tiao, G.~C. and Tan, W.~Y. (1965).
\newblock Bayesian analysis of random-effect models in the analysis of
  variance. i: Posterior distribution of variance components.
\newblock {\em Biometrika}, 52:37--53.

\bibitem[Tiesinga et~al., 2021]{codata2018}
Tiesinga, E., Mohr, P.~J., Newell, D.~B., and Taylor, B.~N. (2021).
\newblock {CODATA} recommended values of the fundamental physical constants:
  2018.
\newblock {\em Reviews of Modern Physics}, to appear.

\bibitem[Toman et~al., 2012]{TomanFischerElster2012}
Toman, B., Fischer, J., and Elster, C. (2012).
\newblock Alternative analyses of measurements of the {P}lanck constant.
\newblock {\em Metrologia}, 49(4):567--571.

\bibitem[Turner et~al., 2015]{Turner2015}
Turner, R.~M., Jackson, D., Wei, Y., Thompson, S.~G., and Higgins, J. (2015).
\newblock Predictive distributions for between-study heterogeneity and simple
  methods for their application in {B}ayesian meta-analysis.
\newblock {\em Statistics in Medicine}, 34(6):984--998.

\bibitem[Veroniki et~al., 2019]{veroniki2019methods}
Veroniki, A.~A., Jackson, D., Bender, R., Kuss, O., Langan, D., Higgins, J.~P.,
  Knapp, G., and Salanti, G. (2019).
\newblock Methods to calculate uncertainty in the estimated overall effect size
  from a random-effects meta-analysis.
\newblock {\em Research synthesis methods}, 10(1):23--43.

\bibitem[Weise and W{\"o}ger, 2000]{weise2000removing}
Weise, K. and W{\"o}ger, W. (2000).
\newblock Removing model and data non-conformity in measurement evaluation.
\newblock {\em Measurement Science and Technology}, 11(12):1649.

\bibitem[Yates and Cochran, 1938]{YatesCochran38}
Yates, F. and Cochran, W.~G. (1938).
\newblock The analysis of groups of experiments.
\newblock {\em Journal of Agricultural Science}, 28:556--580.

\end{thebibliography}
}

\end{document}